\documentclass[journal,twoside]{IEEEtran}
\usepackage{graphicx}
\usepackage{epstopdf}
\usepackage{placeins}
\usepackage{array}
\DeclareGraphicsExtensions{.eps,.pdf} 
\graphicspath{ {figures/} }
\usepackage{cite}
\usepackage{balance}
\usepackage{amssymb,amsmath}
\usepackage{subfigure}
\usepackage[ruled,vlined]{algorithm2e}
\usepackage{amssymb,amsmath,mathtools}
\usepackage{amssymb,amsmath}
\usepackage{algorithmic}
\usepackage{color}
\usepackage{multirow}
\makeatletter
\newcommand\restartchapters{\par
  \setcounter{chapter}{0}%
  \setcounter{section}{0}%
  \gdef\@chapapp{\chaptername}%
  \gdef\thechapter{\@arabic\c@chapter}}
\makeatother

\newtheorem{remark}{\it \underline{Remark}}

\newtheorem{lemma}{\it \underline{Lemma}}

\newtheorem{proposition}{\it \underline{Proposition}}

\newcommand{\tr}{{\mathrm{Tr}}}

\newcommand{\maxi}{{\mathtt{maximize}}}
\newcommand{\mini}{{\mathtt{min}}}

\newcommand{\st}{{\mathtt{s.t.}}}
\newcommand{\F}{{\mathrm{F}}}
\newcommand{\Pro}{{\mathtt{Prob}}}

\renewcommand{\algorithmicrequire}{\textbf{Input:}}

\newcommand{\ds}{\displaystyle}

\makeatletter
\g@addto@macro\normalsize{%
 \setlength\abovedisplayskip{4.0pt}
 \setlength\belowdisplayskip{4.0pt}
 \setlength\abovedisplayshortskip{4.0pt}
 \setlength\belowdisplayshortskip{4.0pt}
}
\makeatother

\begin{document}
\bstctlcite{IEEEexample:BSTcontrol}
\title{A New Design Paradigm for Secure \\ Full-Duplex Multiuser Systems}
\author{
\thanks{Copyright \copyright\; 2018 IEEE. Personal use of this material is permitted. Permission from IEEE must be obtained for all other uses, including reprinting/republishing this material for advertising or promotional purposes, collecting new collected works for resale or redistribution to servers or lists, or reuse of any copyrighted component of this work in other works, by sending a request to pubs-permissions@ieee.org. Citation: V.-D. Nguyen, H. V. Nguyen, O. A. Dobre, and O.-S. Shin, ``A New Design Paradigm for Secure Full-Duplex Multiuser Systems,''  IEEE J. Select. Areas  Commun, accepted to appear, DOI: 10.1109/JSAC.2018.2824379. Available: https://ieeexplore.ieee.org/document/8333690/.}
\IEEEauthorblockN{Van-Dinh Nguyen,  Hieu V. Nguyen,  Octavia A. Dobre, and Oh-Soon Shin}
\thanks{V.-D.~Nguyen, H.~V.~Nguyen, and O.-S.~Shin are with the School of Electronic Engineering $\&$ Department of  ICMC Convergence Technology, Soongsil University, Seoul 06978, Korea  (e-mail: \{nguyenvandinh, hieuvnguyen, osshin\}@ssu.ac.kr).}
\thanks{O.~A.~Dobre is with the Faculty of Engineering and Applied Science, Memorial University, St. John's, NL, Canada (e-mail: odobre@mun.ca).}
\thanks{Part of this work will be presented at the IEEE International Conference on Communications (ICC), USA, May 2018.}
\vspace*{-0.45cm}}
\maketitle
\thispagestyle{empty}
\pagestyle{empty}
\begin{abstract}
We consider a full-duplex (FD) multiuser system where an FD base station (BS) is designed to simultaneously serve both downlink (DL)  and uplink (UL) users in the presence of half-duplex eavesdroppers (Eves). The problem is to maximize the minimum (max-min) secrecy rate (SR) among all legitimate users, where the information signals at the FD-BS are accompanied with  artificial noise  to debilitate the Eves'  channels. To enhance the max-min SR, a major part of the power budget should be allocated to serve the users with poor channel qualities, such as those  far from the FD-BS,  undermining the SR for other users, and thus compromising the  SR per-user.
In addition, the main obstacle in designing an FD system is due to the self-interference (SI)  and  co-channel interference (CCI)  among users. We therefore propose an alternative solution, where the FD-BS uses a fraction of the time block to serve near DL users and far UL users, and the remaining fractional time to serve  other users. The proposed scheme mitigates the harmful effects of SI, CCI and multiuser interference, and  provides system robustness. The SR optimization  problem has a highly nonconcave and nonsmooth objective, subject to nonconvex constraints. For the case of perfect channel state information (CSI), we develop a  low-complexity path-following algorithm, which involves only a simple convex program of moderate dimension at each iteration. We show that our path-following algorithm guarantees convergence at least to a local optimum. Then, we extend the path-following algorithm to the cases of partially known Eves' CSI, where only statistics of CSI for the Eves are known, and  worst-case scenario in which  Eves can employ a more advanced linear decoder. The merit of our proposed approach is further demonstrated  by extensive numerical results.
\end{abstract}

\section{Introduction} \label{Introduction}
The explosive demand for new services and data traffic  is constantly on the rise, pushing  new developments of signal processing and communication technologies \cite{Andrews-14-A}. It is widely believed that multiple-antenna techniques  can offer extra degrees-of-freedom (DoF) to  efficiently allocate resources, which helps reduce the bandwidth and/or power while still maintaining the same quality-of-service (QoS) requirements \cite{MietznerCST09}. However, the half-duplex (HD) radio, where downlink (DL) and uplink (UL) transmissions occur orthogonally   either in time or in frequency,  leads to  under-utilization  of radio resources,  and may no longer provide  substantial improvements in system performance even if multiple antennas are employed.

By enabling simultaneous transmission and reception on the same channel, full-duplex (FD) radio, which  offers considerable potential of doubling the spectral efficiency compared to its HD counterpart, has arisen as a promising technology for the fifth generation of mobile communications (5G) \cite{ZhangCM15,Wong5Gbook17}. The major challenge in designing an FD radio is to suppress the self-interference (SI) caused by  signal leakage from the DL transmission to the UL  reception on the same device to a  potentially suitable level, such as a few dB above background noise. Fortunately, recent advances in hardware design have allowed the FD radio  to be  implemented at a reasonable cost while canceling out  most of the SI \cite{Riihonen-SP-11,DUPLO,Duarte:TWC:12,Saetal14}. Since then, applying FD radio to a base station (BS) in small cell-based systems or to an access point in WiFi, in which the  transmit power is relatively low,   has been widely considered. FD for multiple-input multiple-output (MIMO) in single-cell systems has been investigated to  achieve a higher spectral efficiency \cite{Dan:TWC:14,SH:TWC:15,Nguyen:Access:17}, and an extension to multi-cell scenarios has also been considered in \cite{Tam:TCOM:16,YadavAcess17,AquilinaTCOM17}.  Another downside of FD radio in a typical cellular network is that co-channel interference (CCI) caused by the UL transmission of UL users  severely impairs the DL reception of DL users. Therefore, it is  challenging to fully capitalize on the benefits  that FD radios may bring to 5G wireless networks.

Wireless networks have a very wide range of applications, and an unprecedented amount of personal and sensitive information is transmitted over wireless channels. Consequently,   wireless network security is a crucial issue  due to the unalterable open nature of the wireless medium. Physical-layer (PHY-layer) security can potentially provide  information privacy at the PHY-layer by taking advantage of the characteristics  of the wireless medium \cite{ChenCST16,WuTCOM17,WuIT16,Nguyen:TIFS:16,ZhengTSP13,LiSPL16,AkgunTCOM17,ZhuTSP14,ZhuTWC16,SunTWC16,Lui:SP:14,Ng-14-A}. An effective means to delivering PHY-layer security is to adopt artificial noise (AN) to degrade the decoding capability of the eavesdropper (Eve) \cite{WuTCOM17,WuIT16,Nguyen:TIFS:16}, such that the confidential messages are useless for  Eve. Notably, with FD radio, we can exploit AN even more effectively \cite{ChenCST16}.

\vspace*{-0.4cm}
\subsection{Related Work}
In this subsection, we  discuss the most recent and relevant works for PHY-layer security that exploit FD radio.  Zheng \textit{et al.} \cite{ZhengTSP13} proposed a self-protection scheme by exploiting FD radio at  the desired user to simultaneously receive information and transmit AN in point-to-point transmission. The secrecy rate (SR) optimization was studied in both cases of known and unknown channel state information (CSI) of the Eve. The work in \cite{LiSPL16} considered a  MIMO Gaussian wiretap channel with an FD jamming receiver to secure the DoF. The authors in \cite{AkgunTCOM17} extended the prior work to a multiuser scenario, where the transmitter is equipped with multiple antennas to guarantee the individual SR for multiple  users. However, these approaches require  high hardware complexity of the receivers, which may be difficult to achieve in practice because the end devices should be of very low complexity.

By shifting the computational hardware complexity from the receivers to the BS in cellular networks,   FD-BS secure communications are of  great interest thanks to providing communication secrecy to both UL and DL transmissions. In \cite{ZhuTSP14}, joint information  and AN beamforming at the FD-BS was investigated to guarantee PHY-layer security of  single-antenna UL  and DL users. However, this work assumed that there is no SI and CCI, which  is highly idealistic. Therefore, an extension was proposed in \cite{ZhuTWC16} by considering both SI and CCI. The work in \cite{SunTWC16} analyzed a trade-off between the DL and UL transmit power in FD systems to secure
multiple DL  and UL users. The common technique used in these works is  semi-definite relaxation (SDR) that relaxes  nonconvex constraints  to arrive at a semi-definite  program (SDP). Though in-depth results were presented, the beamforming vectors that are recovered from the covariance matrices of SDR may not perform effectively, and the dimension of covariance matrices is relatively large \cite{PhanTSP}. By leveraging the inner approximation principle,  our previous works in \cite{Nguyen:Access:17} and \cite{Nguyen:TCOM:17} have recently proposed  path-following  algorithms to efficiently address  the nonconvex optimization problems in FD multiuser systems by solving a sequence of  convex programs. These  algorithms completely avoid rank-one constraints and jointly optimize all optimization variables in a single layer.

\vspace*{-0.3cm}	
\subsection{Motivation and Contributions}
Even with recent advances in hardware design for SI cancellation techniques, the harmful effect  of residual SI cannot be neglected if it is not properly controlled, and is proportional to the DL transmission power.  The CCI may become strong and uncontrolled whenever an UL user is located near DL users. These shortcomings limit the performance of FD systems \cite{Dan:TWC:14,SH:TWC:15,Nguyen:Access:17,Tam:TCOM:16,YadavAcess17,AquilinaTCOM17,SunTWC16}. In addition, a major part of the FD-BS power budget allocated to the DL users with poor channel conditions to improve their QoS  will significantly reduce the QoS for other users due to an increase in both the SI and multiuser interference (MUI). Recently, non-orthogonal multiple access (NOMA) \cite{Wong5Gbook17,Islam:CST17,DSP16,Nguyen:JSAC:17} has been introduced to improve the far users' (i.e.,  the cell-edge users) throughput by allowing near users (i.e.,  the cell-center users) to access and decode their intended signals. In other words,  far users must sacrifice their own information privacy in NOMA \cite{Islam:CST17,DSP16,Nguyen:JSAC:17,Choi15,NguyenCLFT17}. In all aforementioned work in FD security, the number of transmit antennas is usually required to be larger than  the number of users to efficiently manage the network interference. Otherwise, the AN and MUI will impair the channel quality of the desired users, leading to a significant loss in system performance.

In this paper, we propose a new transmission design to further resolve the  practical restrictions given above. Specifically, the near DL users and far UL users are served in a fraction of the time block, and then FD-BS uses the remaining  fractional time  to serve far DL users and near UL users. It is worth noting that the effects of SI, CCI and MUI are clearly reduced while the information privacy for far DL users  is preserved (all DL users are allowed
to access and decode their intended signals only). On the other hand,  FD-BS can  effectively perform  transmit beamforming even if the number of DL users exceeds the number of transmit antennas because the number of users that are served at the same time is effectively reduced. There are multiple-antenna Eves that overhear the information signals from both DL and UL channels. We are concerned with the problem of jointly optimizing linear precoders/beamformers at the FD-BS,  UL transmit power allocation and fractional time (FT) to maximize the minimum SR  among all legitimate users subject to  power constraints. In general, such a design problem involves optimization of nonconcave and nonsmooth objective functions subject to nonconvex constraints, for which the optimal solution is computationally difficult to find. Note that SDR cannot be directly applied to such a challenging problem since the optimization problem resulting from the SDR is still highly
nonconvex. This paper is centered on the  inner approximation framework to directly tackle the nonconvexity
of the considered optimization problem. The main contributions of this paper are summarized as follows:
\begin{enumerate}
   \item We propose a new transmission model for FD security to optimize simultaneous DL and UL information  privacy by exploring user grouping-based FT model, which helps manage the  network interference more effectively than aiming to focus the interference at  Eves.
   \item We first assume perfect CSI to realize the potential benefits of our new model, for which a path-following computational procedure is proposed. The core idea behind our approach is to develop a new inner approximation of the nonconvex problem, which guarantees  convergence at least to local optima. The convex program solved at each iteration is of moderate dimension since it does not require rank-constrained optimization,  and thus, its computation is very  efficient.
	\item When only the statistics of CSI (SCSI) for Eves are known, we reformulate the optimization problem by replacing a nonconvex probabilistic constraint with a tractable nonconvex constraint which can be further shaped to have a set of convex constraints.
	\item We determine the optimal solution for a worst-case scenario (WCS) of secure communications, where  Eves adopt a more advanced linear decoder to cancel all multiuser interference.
\item Extensive numerical results show that the proposed algorithms converge quite quickly in a few iterations and greatly improve the SR performance over  existing schemes, i.e., HD, conventional FD and FD-NOMA. It also confirms the robustness of the proposed approach against the significant effects of SI and  DoF bottleneck.
\end{enumerate}


\subsection{Paper Organization and Notation}
The rest of the paper is organized as follows. The system model and problem statement are given in Section~\ref{System Model}.  Path-following algorithms based on a convex approximation for the SR maximization (SRM)  problem with known CSI and statistical CSI of Eves are developed in Section~\ref{sec:knownCSI} and Section~\ref{sec:ECDI}, respectively. Section~\ref{sec:MMSEEavs} is devoted to the computation for the SRM-WCS problem. Numerical results are illustrated in Section~\ref{NumericalResults}, and Section~\ref{Conclusion} concludes the paper.

\emph{Notation}: Lower-case letters, bold lower-case letters and bold upper-case letters represent scalars, vectors and matrices, respectively. $\mathbf{X}^{H}$, $\mathbf{X}^{T}$ and $\mathbf{X}^{*}$   are the Hermitian transpose, normal transpose and conjugate  of a matrix $\mathbf{X}$, respectively. The trace of a matrix $\mathbf{X}$ is denoted by $\tr(\mathbf{X})$. $\|\cdot\|_\F$, $\|\cdot\|$ and $|\cdot|$ denote a matrix's Frobenius  norm,  a vector's Euclidean norm and  absolute value of a complex scalar, respectively. $\mathbf{I}_N$ represents an $N\times N$ identity matrix. $\mathbf{x}\sim\mathcal{CN}(\boldsymbol{\eta},\boldsymbol{Z})$ means that $\mathbf{x}$ is a random vector following a complex circularly symmetric Gaussian distribution with mean  $\boldsymbol{\eta}$ and covariance matrix $\boldsymbol{Z}$. $\mathbb{E}\{\cdot\}$ denotes the statistical expectation. The notation $\mathbf{X}\succeq\mathbf{0}$ ($\mathbf{X}\succ\mathbf{0}$) means that matrix $\mathbf{X}$ is  positive semi-definite (definite).  $\Re\{\cdot\}$ represents the real part of the argument.   $\nabla_{\mathbf{x}}f(\mathbf{x})$ is the gradient of  $f(\mathbf{x})$.

\section{System Model and  Problem Formulation} \label{System Model}


\subsection{Signal Processing Model}
\begin{figure}[t]
\centering
\includegraphics[width=0.49\textwidth,trim={-0cm 0.0cm 0cm 0cm}]{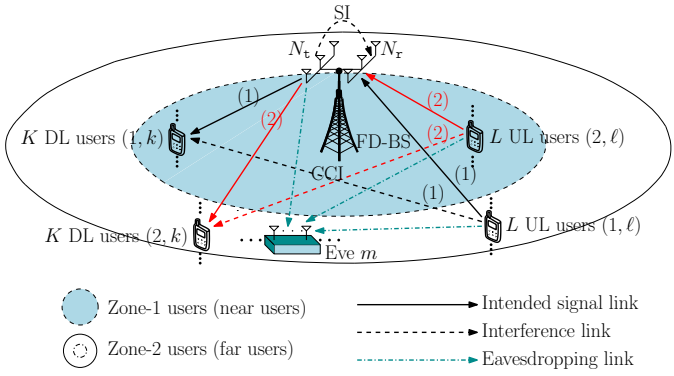}
\caption{A multiuser system model with an FD-BS serving $2K$ DL users and $2L$ UL users in the presence of $M$ Eves.}
\label{fig:SM:1}
\end{figure}

Consider a multiuser communication system illustrated in Fig.~\ref{fig:SM:1}, where the FD-BS is equipped with $N_{\mathtt{t}}$ transmit antennas and $N_{\mathtt{r}}$  receive antennas to simultaneously serve $2K$ DL users and $2L$ UL users, respectively, over the same radio frequency band. Each legitimate user  is equipped with a single antenna and operates in the HD mode to ensure low hardware complexity. Both DL and UL transmissions are overheard by $M$ non-colluding  Eves, where   Eve $m$ has $N_{e,m}$ antennas.\footnote{The scenario can be easily extended to colluding Eves by incorporating $M$ Eves into one with $\sum_{m=1}^MN_{e,m}$ antennas.} Herein, we use the most natural and efficient divisions of the coverage area \cite{DSP16,Nguyen:JSAC:17}. In particular,  users are randomly placed into two zones, such that there are $K$ DL users and $L$ UL users located in a zone nearer the FD-BS (referred to as zone-1 of near users), and $K$ DL users and $L$ UL users are located in a zone farther from the FD-BS (called  zone-2 of far users) \cite{DSP16}.  Note that our proposed
algorithms can be further adjusted to the case of  different  numbers of users located in each zone.  

\begin{figure}[t]
\centering
\includegraphics[width=0.48\textwidth]{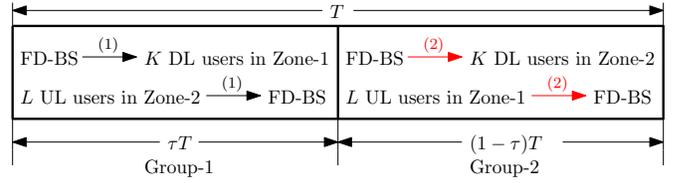}
\caption{Time slot structure to serve $K$ near DL users and $L$ far UL users within the duration  $\tau T$, as well as $K$ far DL users and $L$ near UL users in the remaining duration $(1 - \tau) T$.}
\label{fig:TDMA}
\end{figure}

In this paper, we split each communication time block, denoted by $T$, into two sub-time blocks orthogonally, as shown in Fig.~\ref{fig:TDMA}. As previously mentioned, in order to mitigate the harmful effects of  SI and CCI,   far DL users (near DL users) and  far UL users (near UL users) should be scheduled in a different time slot. The FD-BS serves the DL users (UL users) with similar channel conditions, which helps reduce the MUI. As a consequence, $K$ near DL users and $L$ far UL users are grouped into   group-1 and are served in the first duration  $\tau T\ (0 < \tau < 1)$, while $K$ far DL users and $L$ near UL users are grouped into group-2 and are served in the remaining duration  $(1-\tau)T$.   Although each group still operates in the FD mode, the inter-group interference, i.e., interference across groups 1 and 2, is perfectly eliminated thanks to the FT allocation. This is in contrast with the conventional FD systems which simultaneously serves all UL and DL users. Without loss of generality, the communication time block $T$ is normalized to 1. Upon denoting $\mathcal{K}\triangleq\{1,2,\cdots,K\}$ and $\mathcal{L}\triangleq\{1,2,\cdots,L\}$,  the sets of DL  and UL users are $\mathcal{D}\triangleq \mathcal{I}\times\mathcal{K}$ and $\mathcal{U}\triangleq \mathcal{I}\times\mathcal{L}$ for $\mathcal{I}\triangleq\{1,2\}$, respectively. Thus, the $k$-th DL user and the $\ell$-th UL user in the $i$-th group are referred to as DL user $(i,k)\in\mathcal{D}$ and UL user $(i,\ell)\in\mathcal{U}$, respectively.

\textit{1) Received Signal Model at the FD-BS and DL Users:}
We consider that the FD-BS deploys a transmit beamformer $\mathbf{w}_{i,k}\in\mathbb{C}^{N_{\mathtt{t}}\times 1}$ to transfer the information bearing signal $x_{i,k}$ with $\mathbb{E}\{|x_{i,k}|^2\}=1$  to DL user $(i,k)$. 
Since all UL users have only a single antenna,  they are unable to generate  AN for jamming. To guarantee secure communication in both DL and UL channels,  FD-BS also injects  AN signals to interfere with the reception of the Eves. Hence, the DL transmit signals at the FD-BS intended for $K$ DL users in zone-$i$  can be expressed as
\begin{equation} 
\mathbf{x}_i = \sum_{k=1}^K\mathbf{w}_{i,k}x_{i,k} + \mathbf{v}_i,\ \forall i\in\mathcal{I}
\end{equation}
 where $\mathbf{v}_i\in\mathbb{C}^{N_{\mathtt{t}}\times 1}$ is the AN vector whose elements are zero-mean complex Gaussian random variables, i.e., $\mathbf{v}_i\sim\mathcal{CN}(\mathbf{0},\mathbf{V}_i\mathbf{V}_i^H)$ with $\mathbf{V}_i\in\mathbb{C}^{N_{\mathtt{t}}\times N_{\mathtt{t}}}$.  All channels are assumed to follow frequency-flat fading, which accounts for the effects of both large-scale path loss and small-scale fading. The received signal at DL user $(i,k)$  can be expressed as
\begin{IEEEeqnarray}{rCl}\label{eq:signalDL1k}
y_{i,k} = \mathbf{h}_{i,k}^H\mathbf{w}_{i,k}x_{i,k} + \sum\nolimits_{j=1,j\neq k}^K\mathbf{h}_{i,k}^H\mathbf{w}_{i,j}x_{i,j} \nonumber\\
+\; \mathbf{h}_{i,k}^H\mathbf{v}_{i} + \sum\nolimits_{\ell=1}^Lf_{i,k,\ell}\rho_{i,\ell}\tilde{x}_{i,\ell} + n_{i,k}
\end{IEEEeqnarray}
where $\mathbf{h}_{i,k}\in\mathbb{C}^{N_{\mathtt{t}}\times 1}$ is the transmit channel vector from the FD-BS to DL user $(i,k)$. In \eqref{eq:signalDL1k}, the term $\sum_{\ell=1}^Lf_{i,k,\ell}\rho_{i,\ell}\tilde{x}_{i,\ell}$ represents the CCI from $L$ UL users to DL user $(i,k)$, where $f_{i,k,\ell}\in\mathbb{C}$, $\rho_{i,\ell}$ and $\tilde{x}_{i,\ell}$ with $\mathbb{E}\{|\tilde{x}_{i,\ell}|^2\}=1$ are the complex channel coefficient from UL user $(i,\ell)$ to DL user $(i,k)$, transmit power and message of UL user $(i,\ell)$, respectively.  $n_{i,k}\sim\mathcal{CN}(0,\sigma^2)$  denotes the additive white Gaussian noise (AWGN) at DL user $(i,k)$. By defining $\tau_1:=\tau$ and $\tau_2:=1-\tau$, the information rate decoded by DL user $(i,k)$  in nats/sec/Hz is given by \cite{Nguyen:TCOM:17} 
\begin{IEEEeqnarray}{rCl}\label{eq:RateDLUs}
C_{i,k}^{\mathtt{D}}(\mathbf{X}_i,\tau_i) &=& \tau_i\ln\Bigl(1 + \frac{|\mathbf{h}_{i,k}^H\mathbf{w}_{i,k}|^2}{\varphi_{i,k}(\mathbf{X}_i)}\Bigr)
\end{IEEEeqnarray}
where $\mathbf{X}_i\triangleq  \bigl\{\mathbf{w}_{i},\mathbf{V}_i,\boldsymbol{\rho}_i\bigl\},$ with $\mathbf{w}_{i}\triangleq\{\mathbf{w}_{i,k}\}_{k\in\mathcal{K}}$ and $\boldsymbol{\rho}_i\triangleq\{\rho_{i,\ell}\}_{\ell\in\mathcal{L}}, i=1,2,$  are the matrix encompassing the beamformers/precoders in the DL and transmit power allocation of all users in the UL in the $i$-th group, and 
\begin{IEEEeqnarray}{rCl}\label{eq:RateDLUsDominator}
\varphi_{i,k}(\mathbf{X}_i) \triangleq \sum\nolimits_{j=1,j\neq k}^K|\mathbf{h}_{i,k}^H\mathbf{w}_{i,j}|^2 + \|\mathbf{h}_{i,k}^H\mathbf{V}_{i}\|^2 \nonumber\\
+\; \sum\nolimits_{\ell=1}^L\rho_{i,\ell}^2|f_{i,k,\ell}|^2 + \sigma^2.\nonumber
\end{IEEEeqnarray}

The received signal at  the FD-BS for reception of $L$ UL users in the $i$-th group can be expressed as
\begin{IEEEeqnarray}{rCl}\label{eq:signalUL2p}
\mathbf{y}_{i,bs} =  \sum\nolimits_{\ell=1}^L\rho_{i,\ell}\mathbf{g}_{i,\ell}\tilde{x}_{i,\ell} + \sqrt{\sigma_{\mathtt{SI}}}\sum\nolimits_{k=1}^K\mathbf{G}_{\mathtt{SI}}^H\mathbf{w}_{i,k}x_{i,k}\nonumber\\
 +\; \sqrt{\sigma_{\mathtt{SI}}}\mathbf{G}_{\mathtt{SI}}^H\mathbf{v}_{i}  + \mathbf{n}_{i,bs}\quad
\end{IEEEeqnarray}
where $\mathbf{g}_{i,\ell}\in\mathbb{C}^{N_{\mathtt{r}}\times 1}$ is the receive channel vector from UL user $(i,\ell)$ to the FD-BS. The term $\sqrt{\sigma_{\mathtt{SI}}}\sum_{k=1}^K\mathbf{G}_{\mathtt{SI}}^H\mathbf{w}_{i,k}x_{i,k}$ in (\ref{eq:signalUL2p}) represents the residual SI  after all real-time cancellation in analog and digital domains
\cite{Saetal14}; $\mathbf{G}_{\mathtt{SI}}\in\mathbb{C}^{N_{\mathtt{t}}\times N_{\mathtt{r}}}$ denotes a fading loop channel which impairs the UL signal detection at the FD-BS due to the concurrent DL transmission  and $0\leq\sigma_{\mathtt{SI}}< 1$ is used to model the degree of residual SI propagation \cite{Riihonen-SP-11}. $\mathbf{n}_{i,bs}\sim\mathcal{CN}(\mathbf{0},\sigma^2\mathbf{I}_{N_{\mathtt{r}}})$ denotes the AWGN at  the FD-BS. To maximize the information rates of UL users, we adopt the minimum mean square error and successive interference
cancellation (MMSE-SIC) decoder at  the FD-BS \cite{Tse:book:05}. For $L$ UL users in each group, we assume that the decoding order follows the UL users' index, i.e., $\ell =1,2,\cdots, L$, with the  Foschini ordering. In other words, the strongest signal is decoded first, while weakest signal is decoded last to support the most vulnerable UL users. Hence, the information rate in decoding UL user $(i,\ell)$'s  message is given by \cite{Nguyen:TCOM:17} 
\begin{IEEEeqnarray}{rCl}\label{eq:RateULUs}
C_{i,\ell}^{\mathtt{U}}(\mathbf{X}_i,\tau_i) &=& \tau_i\ln\Bigl(1 + \rho_{i,\ell}^2\mathbf{g}_{i,\ell}^H\boldsymbol{\Phi}_{i,\ell}(\mathbf{X}_i)^{-1}\mathbf{g}_{i,\ell}  \Bigr)
\end{IEEEeqnarray}
where 
\begin{IEEEeqnarray}{rCl}\label{eq:RateULUsDominator}
\boldsymbol{\Phi}_{i,\ell}(\mathbf{X}_i) \triangleq \sum\nolimits_{j>\ell}^L\rho_{i,j}^2\mathbf{g}_{i,j}\mathbf{g}_{i,j}^H + \sigma_{\mathtt{SI}}\sum\nolimits_{k=1}^K\mathbf{G}_{\mathtt{SI}}^H\mathbf{w}_{i,k}\mathbf{w}_{i,k}^H\mathbf{G}_{\mathtt{SI}} \nonumber\\
+\; \sigma_{\mathtt{SI}}\mathbf{G}_{\mathtt{SI}}^H\mathbf{V}_{i}\mathbf{V}_{i}^H\mathbf{G}_{\mathtt{SI}}  + \sigma^2\mathbf{I}_{N_{\mathtt{r}}}.\nonumber
\end{IEEEeqnarray}

\textit{2) Received Signal Model at Eves:} After performing handshaking with the FD-BS, we assume that the Eves are also aware of the FT $\tau_i$.
The information signals of group-$i$ leaked out to the $m$-th Eve during the FT $\tau_i$ can be expressed as\footnote{Note that if the Eves are not aware of the FT $\tau_i$, the  received signals at Eve $m$ in \eqref{eq:Eavsignal} will include  the inter-group interference, which leads to a lower bound on the information rate of the Eves. Such a design would be unfair to the Eves, and therefore, we do not pursue this here. }
\begin{IEEEeqnarray}{rCl}\label{eq:Eavsignal}
\mathbf{y}_{i,m} = \mathbf{H}_m^H\Bigl(\sum\nolimits_{k=1}^K\mathbf{w}_{i,k}x_{i,k}+ \mathbf{v}_i\Bigr) \nonumber\\
+ \sum\nolimits_{\ell=1}^L\rho_{i,\ell}\mathbf{g}_{m,i,\ell}^H\tilde{x}_{i,\ell} + \mathbf{n}_{e,m}
\end{IEEEeqnarray}
where $\mathbf{H}_m\in\mathbb{C}^{N_{\mathtt{t}}\times N_{e,m}}$ and $\mathbf{g}_{m,i,\ell}\in\mathbb{C}^{1\times N_{e,m}}$ are the wiretap channel matrix and vector from the FD-BS and UL user ($i,\ell$) to  the $m$-th Eve, respectively. $\mathbf{n}_{e,m}\sim\mathcal{CN}(\mathbf{0},\sigma^2\mathbf{I}_{N_{e,m}})$ denotes the AWGN at  Eve $m$. The worst-case information (WCI) rates at  the $m$-th Eve, corresponding to the signals targeted for DL user $(i,k)$ and UL user $(i,\ell)$, are given by
\begin{IEEEeqnarray}{rCl}\label{eq:RateEvam}
C_{m,i,k}^{\mathtt{ED}}(\mathbf{X}_i,\tau_i) &=& \tau_i\ln\Bigl(1 + \frac{\|\mathbf{H}_{m}^H\mathbf{w}_{i,k}\|^2}{\psi_{m,i,k}(\mathbf{X}_i)}\Bigr),\IEEEyessubnumber\\
C_{m,i,\ell}^{\mathtt{EU}}(\mathbf{X}_i,\tau_i) &=& \tau_i\ln\Bigl(1 + \frac{\rho_{i,\ell}^2\|\mathbf{g}_{m,i,\ell}^H\|^2}{\chi_{m,i,\ell}(\mathbf{X}_i)}\Bigr)\IEEEyessubnumber
\end{IEEEeqnarray}
respectively, where
\begin{IEEEeqnarray}{rCl}\label{eq:EavDominator}
\psi_{m,i,k}(\mathbf{X}_i) &\triangleq& \sum\nolimits_{j=1,j\neq k}^K\|\mathbf{H}_{m}^H\mathbf{w}_{i,j}\|^2  + \|\mathbf{H}_{m}^H\mathbf{V}_{i}\|_\F^2 \nonumber\\&+& \sum\nolimits_{\ell=1}^L\rho_{i,\ell}^2\|\mathbf{g}_{m,i,\ell}^H\|^2 + N_{e,m}\sigma^2,   \nonumber\\
\chi_{m,i,\ell}(\mathbf{X}_i) &\triangleq& \sum\nolimits_{k=1}^K\|\mathbf{H}_{m}^H\mathbf{w}_{i,k}\|^2 + \|\mathbf{H}_{m}^H\mathbf{V}_{i}\|_\F^2 \nonumber\\ 
 &+& \sum\nolimits_{j=1,j\neq \ell}^L \rho_{i,j}^2\|\mathbf{g}_{m,i,j}^H\|^2 + N_{e,m}\sigma^2.\nonumber
\end{IEEEeqnarray}

\subsection{Optimization Problem Formulation} \label{OptimizationProblem}
The channel of each legitimate user together with $M$ Eves form a compound wiretap channel for which the SR expressions of DL user $(i,k)$ and UL user $(i,\ell)$ can be expressed as \cite{Lian:EUR:09,SunTWC16}
\begin{IEEEeqnarray}{rCl}\label{eq:Secrecyrates}
R_{i,k}^{\mathtt{D}}(\mathbf{X}_i,\tau_i) &\triangleq& \bigl[C_{i,k}^{\mathtt{D}}(\mathbf{X}_i,\tau_i) - \underset{m\in\mathcal{M}}{\max}\,C_{m,i,k}^{\mathtt{ED}}(\mathbf{X}_i,\tau_i)\bigr]^+,\quad\  \IEEEyessubnumber\\
R_{i,\ell}^{\mathtt{U}}(\mathbf{X}_i,\tau_i) &\triangleq& \bigl[C_{i,\ell}^{\mathtt{U}}(\mathbf{X}_i,\tau_i) - \underset{m\in\mathcal{M}}{\max}\,C_{m,i,\ell}^{\mathtt{EU}}(\mathbf{X}_i,\tau_i)\bigl]^+\IEEEyessubnumber
\end{IEEEeqnarray}
respectively, where $\mathcal{M}\triangleq\{1,2,\cdots,M\}$ and $[x]^+\triangleq\max\{0,x\}$.

We aim to jointly optimize  the transmit information vectors, AN matrices $(\{\mathbf{X}_i\}_{i\in\mathcal{I}}$) and  the FT  $(\{\tau_i\}_{i\in\mathcal{I}})$ to maximize the minimum (max-min) SR among all legitimate users. The SRM problem with Eves' WCI rate, referred to as  SRM-EWCI for short, can be mathematically formulated as 
\begin{IEEEeqnarray}{rCl}\label{eq:OP1}
&&\underset{\mathbf{X},\boldsymbol{\tau}}{\maxi}\;\underset{\substack{(i,k)\in\mathcal{D}\\ (i,\ell)\in\mathcal{U}}}{\mini}\;\left\{R_{i,k}^{\mathtt{D}}(\mathbf{X}_i,\tau_i),R_{i,\ell}^{\mathtt{U}}(\mathbf{X}_i,\tau_i)\right\} \IEEEyessubnumber\label{eq:OP1:a}\\
&&\st\   \sum_{i=1}^2\tau_i\Bigl(\sum_{k=1}^K\|\mathbf{w}_{i,k}\|^2 + \|\mathbf{V}_i\|_\F^2\Bigr) \leq  P_{bs}^{\max},   \IEEEyessubnumber\label{eq:OP1:b} \\
&&\qquad \tau_1\rho_{1,\ell}^2 \leq P_{1,\ell}^{\max},\;\forall \ell\in\mathcal{L},  \IEEEyessubnumber\label{eq:OP1:c}\\
&&\qquad \tau_2 \rho_{2,\ell}^2 \leq P_{2,\ell}^{\max},\;\forall \ell\in\mathcal{L},  \IEEEyessubnumber\label{eq:OP1:d}\\
&&\qquad \rho_{1,\ell} \geq 0, \rho_{2,\ell} \geq 0,\,\forall \ell\in\mathcal{L}, \IEEEyessubnumber\label{eq:OP1:e}\\
&&\qquad \tau_1>0, \tau_2>0, \tau_1 + \tau_2 \leq 1\IEEEyessubnumber\label{eq:OP1:f}
\end{IEEEeqnarray}
where $\mathbf{X} \triangleq \{\mathbf{X}_1, \mathbf{X}_2\}$ and $\boldsymbol{\tau}\triangleq\{\tau_1,\tau_2\}$. Constraint \eqref{eq:OP1:b} means that the total transmit power at the FD-BS, which is allocated across different time fractions, does not exceed the  power budget, $P_{bs}^{\max}$, while  constraints \eqref{eq:OP1:c} and \eqref{eq:OP1:d} are  individual transmit power at UL user ($i,\ell$) in its service time, with $P_{i,\ell}^{\max}$ being the power budget (see e.g., \cite{Nguyen:TCOM:17,Nguyen:Access:17,Yang:TVT:17} for these realistic power constraints). Finding an optimal solution to the SRM-EWCI problem \eqref{eq:OP1} is challenging because the objective \eqref{eq:OP1:a} is  nonconcave and  constraints \eqref{eq:OP1:b}-\eqref{eq:OP1:d} are also nonconvex due to coupling between $\mathbf{X}$ and $\boldsymbol{\tau}$. 

\begin{remark}\label{Remark1} In a practical scenario, the DL and UL traffic demands in current generation wireless networks are typically asymmetric.
Another optimization problem of interest is to maximize the minimum SR of DL users subject to the SR constraints of UL users as follows:
\begin{IEEEeqnarray}{rCl}\label{eq:OP1a}
\underset{\mathbf{X},\boldsymbol{\tau}}{\maxi}&&\;\underset{(i,k)\in\mathcal{D}}{\mini}\;\left\{R_{i,k}^{\mathtt{D}}(\mathbf{X}_i,\tau_i)\right\} \IEEEyessubnumber\label{eq:OP1a:a}\\
\st\ && \eqref{eq:OP1:b}-\eqref{eq:OP1:f},\IEEEyessubnumber\label{eq:OP1a:c}\\
&&  R_{i,\ell}^{\mathtt{U}}(\mathbf{X}_i,\tau_i) \geq \bar{\mathtt{R}}_{i,\ell}^{\mathtt{U}},  \forall(i,\ell)\in\mathcal{U} \IEEEyessubnumber\label{eq:OP1a:b}
\end{IEEEeqnarray}
where the QoS constraints in \eqref{eq:OP1a:b} set a minimum SR requirement $\bar{\mathtt{R}}_{i,\ell}^{\mathtt{U}}$ at  UL user $(i,\ell)$.  It should be emphasized that the systematic approach in this paper is expected to
be applicable for such a problem (this will be elaborated in Section \ref{NumericalResults}).
\end{remark}

\section{Proposed Method with Known CSI}\label{sec:knownCSI}

In this section, the CSI of the  users (including Eves)   is assumed to be perfectly known at the transmitters. Channel reciprocity of  UL and DL channels in  time  division duplex (TDD) mode can be adopted for small cell systems
as those considered in this paper. The channels for all users (with a low degree of mobility) can be acquired at the FD-BS by requesting them to send  pilot signals to the FD-BS, and thus  these estimated channels can be assumed to be perfectly available \cite{Nguyen:TCOM:17,Dan:TWC:14}.
Likewise, the CSI between an  UL user and the receivers (DL users and Eves) can be estimated through TDD since any transmitted signal includes short-training and long-training sequences (e.g., a part of preamble). This way, any UL user can overhear and estimate the channels from DL users and Eves, and then these estimated channels can be acquired  at the  FD-BS  by polling each UL user.  After  CSI acquisition, we assume that only $2K$ DL users and $2L$ UL users are scheduled to be simultaneously served as in IEEE 802.11ac. Herein, $M$ unscheduled users are not necessarily malicious, but are untrusted users. Thus, $M$  unscheduled users are treated as eavesdroppers but with perfectly known CSI.

\subsection{Equivalent Transformations for \eqref{eq:OP1}} 
To solve the max-min SR problem in \eqref{eq:OP1}, we  present a path-following algorithm under which each iteration invokes only a simple convex program of low computational complexity. Toward a tractable form, several proper transformations need to be invoked. Let us start by expressing  \eqref{eq:OP1} equivalently as
\begin{IEEEeqnarray}{rCl}\label{eq:OP2}
\underset{\mathbf{X},\boldsymbol{\tau},\eta}{\maxi}&&\quad \eta \IEEEyessubnumber\label{eq:OP2:a}\\
\st\;&&\eqref{eq:OP1:b}, \eqref{eq:OP1:c}, \eqref{eq:OP1:d}, \eqref{eq:OP1:e}, \eqref{eq:OP1:f}\IEEEyessubnumber\label{eq:OP2:f},\\
&&  R_{i,k}^{\mathtt{D}}(\mathbf{X}_i,\tau_i) \geq \eta,\ \forall(i,k)\in\mathcal{D},   \IEEEyessubnumber\label{eq:OP2:b}\\
&& R_{i,\ell}^{\mathtt{U}}(\mathbf{X}_i,\tau_i) \geq \eta,\ \forall(i,\ell)\in\mathcal{U}  \IEEEyessubnumber\label{eq:OP2:d}.
\end{IEEEeqnarray}
where $\eta$ is an additional variable  to achieve the SR fairness among all DL and UL users. Note that the equivalence between \eqref{eq:OP1} and \eqref{eq:OP2} can be readily verified by checking that  constraints \eqref{eq:OP2:b}-\eqref{eq:OP2:d} must hold with equality at optimum. We now provide a sketch of the proof to verify \eqref{eq:OP2:b}, and  other constraints  follow immediately.
Suppose that $R_{i,k}^{\mathtt{D}}(\mathbf{X}_i,\tau_i) > \eta$ for some $(i,k)$.  Then, there may exist a positive constant $\Delta \eta > 0$ to satisfy $R_{i,k}^{\mathtt{D}}(\mathbf{X}_i,\tau_i) = \eta + \Delta \eta$. As a consequence, $\eta + \Delta \eta$ is also feasible for \eqref{eq:OP2} but yielding a strictly larger objective, and thus, this is a contradiction with the optimality assumption. 
By observing that the objective function \eqref{eq:OP2:a} is monotonic in its argument, the main difficulty in solving \eqref{eq:OP2} is due to the nonconvex constraints \eqref{eq:OP2:b} and \eqref{eq:OP2:d}. To provide a minorant of the SR, we further rewrite \eqref{eq:OP2} as follows: 
\begin{IEEEeqnarray}{rCl}\label{eq:OP3}
\underset{\mathbf{X},\boldsymbol{\tau},\eta,\boldsymbol{\Gamma}}{\maxi}&&\quad \eta \IEEEyessubnumber\label{eq:OP3:a}\\
\st\;&& \eqref{eq:OP1:b}, \eqref{eq:OP1:c}, \eqref{eq:OP1:d}, \eqref{eq:OP1:e}, \eqref{eq:OP1:f}\IEEEyessubnumber\label{eq:OP3:j},\\
&& C_{i,k}^{\mathtt{D}}(\mathbf{X}_i,\tau_i) - \Gamma_{i,k}^{\mathtt{D}}\geq \eta,\ \forall(i,k)\in\mathcal{D},   \IEEEyessubnumber\label{eq:OP3:b}\\
&&C_{m,i,k}^{\mathtt{ED}}(\mathbf{X}_i,\tau_i) \leq \Gamma_{i,k}^{\mathtt{D}},\ \forall m\in\mathcal{M}, (i,k)\in\mathcal{D},\IEEEyessubnumber\label{eq:OP3:c}\qquad\\
&&  C_{i,\ell}^{\mathtt{U}}(\mathbf{X}_i,\tau_i) - \Gamma_{i,\ell}^{\mathtt{U}}\geq \eta,\ \forall(i,\ell)\in\mathcal{U},   \IEEEyessubnumber\label{eq:OP3:f}\\
&&C_{m,i,\ell}^{\mathtt{EU}}(\mathbf{X}_i,\tau_i) \leq \Gamma_{i,\ell}^{\mathtt{U}},\ \forall m\in\mathcal{M}, (i,\ell)\in\mathcal{U}\IEEEyessubnumber\label{eq:OP3:g}
\end{IEEEeqnarray}
where $\boldsymbol{\Gamma}\triangleq\bigr\{\Gamma_{i,k}^{\mathtt{D}},\Gamma_{i,\ell}^{\mathtt{U}}\bigr\}_{i\in\mathcal{I},k\in\mathcal{K},\ell\in\mathcal{L}}$ are newly introduced variables to  tackle the maximum allowable rate of the Eves \cite{Lui:SP:14}. 
 
However,  problem \eqref{eq:OP3}  still remains intractable since it is not amendable to a direct application of the inner approximation method. To this end, we make the variable change:
\begin{equation}\label{eq:changevariables:a}
  \tau_1 = \frac{1}{\alpha_1}\ \text{and}\  \tau_2 = \frac{1}{\alpha_2}
\end{equation}
to equivalently rewrite \eqref{eq:OP1:f} by the following convex constraint
\begin{equation}\label{eq:changevariables:b}
  \frac{1}{\alpha_1} +  \frac{1}{\alpha_2} \leq 1, \forall \alpha_i > 1, i\in\mathcal{I}
\end{equation}
where $\boldsymbol{\alpha}\triangleq\{\alpha_1,\alpha_2\}$ are new variables. Using \eqref{eq:changevariables:a}, constraints \eqref{eq:OP3:b} and \eqref{eq:OP3:f} become
\begin{IEEEeqnarray}{rCl}\label{eq:LegitimateUEs}
&&C_{i,k}^{\mathtt{D}}(\mathbf{X}_i,\alpha_i) \triangleq \frac{\ln\Bigl(1 + \frac{|\mathbf{h}_{i,k}^H\mathbf{w}_{i,k}|^2}{\varphi_{i,k}(\mathbf{X}_i)}\Bigr)}{\alpha_i} \geq \eta + \Gamma_{i,k}^{\mathtt{D}},\forall i,k,\IEEEyessubnumber\label{eq:changeOP3:a}\\
&&C_{i,\ell}^{\mathtt{U}}(\mathbf{X}_i,\alpha_i) \triangleq \nonumber\\
&&\qquad\frac{\ln\Bigl(1 + \rho_{i,\ell}^2\mathbf{g}_{i,\ell}^H\boldsymbol{\Phi}_{i,\ell}(\mathbf{X}_i)^{-1}\mathbf{g}_{i,\ell}  \Bigr)}{\alpha_i}\geq \eta + \Gamma_{i,\ell}^{\mathtt{U}}, \forall i,\ell.\IEEEyessubnumber\label{eq:changeOP3:c}\qquad\,
\end{IEEEeqnarray}
Analogously,  constraints \eqref{eq:OP3:c} and \eqref{eq:OP3:g}  become
\begin{IEEEeqnarray}{rCl}\label{eq:RateEvamChange}
&&C_{m,i,k}^{\mathtt{ED}}(\mathbf{X}_i,\alpha_i) \triangleq \frac{\ln\Bigl(1 + \frac{\|\mathbf{H}_{m}^H\mathbf{w}_{i,k}\|^2}{\psi_{m,i,k}(\mathbf{X}_i)}\Bigr)}{\alpha_i} \leq \Gamma_{i,k}^{\mathtt{D}}, \forall m,i,k,\; \IEEEyessubnumber\label{eq:RateEvamChange:a}\qquad\\
&&C_{m,i,\ell}^{\mathtt{EU}}(\mathbf{X}_i,\alpha_i) \triangleq
 \frac{\ln\Bigl(1 + \frac{\rho_{i,\ell}^2\|\mathbf{g}_{m,i,\ell}^H\|^2}{\chi_{m,i,\ell}(\mathbf{X}_i)}\Bigr)}{\alpha_i}\leq \Gamma_{i,\ell}^{\mathtt{U}}, \forall m,i,\ell.\IEEEyessubnumber\label{eq:RateEvamChange:c}
\end{IEEEeqnarray}

From \eqref{eq:LegitimateUEs} and \eqref{eq:RateEvamChange}, and substituting \eqref{eq:changevariables:a} and \eqref{eq:changevariables:b} into \eqref{eq:OP1:b}-\eqref{eq:OP1:d}, the optimization problem \eqref{eq:OP3} is re-expressed as
\begin{IEEEeqnarray}{rCl}\label{eq:OP4}
\underset{\mathbf{X},\eta,\boldsymbol{\Gamma},\boldsymbol{\alpha}}{\maxi}&&\quad \eta \IEEEyessubnumber\label{eq:OP4:a}\\
\st\;&&  \eqref{eq:OP1:e}, \eqref{eq:changevariables:b}, \eqref{eq:LegitimateUEs}, \eqref{eq:RateEvamChange},  \IEEEyessubnumber\label{eq:OP4:b}\\
&& \Bigl(1-\frac{1}{\alpha_2}\Bigr)\Bigl(\sum_{k=1}^K\|\mathbf{w}_{1,k}\|^2 + \|\mathbf{V}_1\|_\F^2\Bigr) \nonumber\\
&&\  +\; \frac{1}{\alpha_2}\Bigl(\sum_{k=1}^K\|\mathbf{w}_{2,k}\|^2 + \|\mathbf{V}_2\|_\F^2\Bigr) \leq  P_{bs}^{\max},   \IEEEyessubnumber\label{eq:OP4:c}\qquad\\
&& \Bigl(1-\frac{1}{\alpha_2}\Bigl)\rho_{1,\ell}^2 \leq P_{1,\ell}^{\max},\;\forall \ell\in\mathcal{L},  \IEEEyessubnumber\label{eq:OP4:d}\\
&& \frac{1}{\alpha_2}\rho_{2,\ell}^2 \leq P_{2,\ell}^{\max},\;\forall \ell\in\mathcal{L}.  \IEEEyessubnumber\label{eq:OP4:e}
\end{IEEEeqnarray}
Notice that $1/\alpha_1 + 1/\alpha_2 =1$ must hold at optimum, which means that the optimization problem \eqref{eq:OP4} is equivalent to \eqref{eq:OP3}, and the optimal solution for $\boldsymbol{\tau}$ is recovered by \eqref{eq:changevariables:a}.

\subsection{Proposed Convex Approximation-Based Path-Following Method}
We are now in a position to approximate the equivalent formulation in \eqref{eq:OP4}. Note that except for \eqref{eq:OP1:e}, \eqref{eq:changevariables:b} and \eqref{eq:OP4:e}, the rest of the constraints are  nonconvex. The proposed algorithm is mainly based on an inner approximation framework \cite{Marks:78} under which the nonconvex parts are completely exposed.

\textit{Approximation of  Constraints \eqref{eq:LegitimateUEs}:} 
 To develop a convex approximation, we first introduce the following approximation of  function $\zeta(\gamma,t)\triangleq\ln(1+\gamma)/t$ at a feasible point $(\gamma^{(\kappa)},t^{(\kappa)})$:
\begin{eqnarray}\label{eq:ineupper}
\ds\zeta(\gamma,t)&\geq&\ds \mathtt{A}^{(\kappa)} - \mathtt{B}^{(\kappa)}\frac{1}{\gamma} - \mathtt{C}^{(\kappa)}t\label{inq1},\nonumber\\
&&\forall \gamma > 0, \gamma^{(\kappa)} > 0, t>0, t^{(\kappa)}>0
\end{eqnarray}
where
\begin{IEEEeqnarray}{rCl} 
\mathtt{A}^{(\kappa)}&\triangleq& 2\zeta(\gamma^{(\kappa)},t^{(\kappa)}) + \frac{\gamma^{(\kappa)}}{t^{(\kappa)}(\gamma^{(\kappa)}+1)}, \nonumber\\
\mathtt{B}^{(\kappa)}&\triangleq&\frac{(\gamma^{(\kappa)})^2}{t^{(\kappa)}(\gamma^{(\kappa)}+1)},\
\mathtt{C}^{(\kappa)}\triangleq\frac{\zeta(\gamma^{(\kappa)},t^{(\kappa)})}{t^{(\kappa)}}.\nonumber
 \end{IEEEeqnarray}
 The proof of \eqref{eq:ineupper} is given in  Appendix A.   In the spirit of \cite{WES06}, for $\bar{\mathbf{w}}_{i,k}=e^{-j\mathtt{arg}(\mathbf{h}_{i,k}^H\mathbf{w}_{i,k})}\mathbf{w}_{i,k}$ with $j=\sqrt{-1}$, it follows that $|\mathbf{h}_{i,k}^H\mathbf{w}_{i,k}|=\mathbf{h}_{i,k}^H\bar{\mathbf{w}}_{i,k}=\Re\{\mathbf{h}_{i,k}^H\bar{\mathbf{w}}_{i,k}\}\geq 0$ and $|\mathbf{h}_{i',k'}^H\mathbf{w}_{i,k}|=|\mathbf{h}_{i',k'}^H\bar{\mathbf{w}}_{i,k}|$ for all $(i',k')\neq(i,k)$. Thus, $\gamma_{i,k}^{\mathtt{D}}(\mathbf{X}_i)\triangleq|\mathbf{h}_{i,k}^H\mathbf{w}_{i,k}|^2/\varphi_{i,k}(\mathbf{X}_i)$ can be equivalently replaced by
\begin{equation}
\gamma_{i,k}^{\mathtt{D}}(\mathbf{X}_i)=\frac{\bigr(\Re\{\mathbf{h}_{i,k}^H\mathbf{w}_{i,k}\}\bigl)^2}{\varphi_{i,k}(\mathbf{X}_i)},\ \forall(i,k)\in\mathcal{D}
\end{equation}
with the condition 
\begin{IEEEeqnarray}{rCl}\label{eq:poscondi}
\Re\{\mathbf{h}_{i,k}^H\mathbf{w}_{i,k}\}\geq 0,\ \forall (i,k)\in\mathcal{D}. 
 \end{IEEEeqnarray}
By using \eqref{eq:ineupper}, at  a feasible point $(\mathbf{X}_i^{(\kappa)},\alpha_i^{(\kappa)})$  found at the $(\kappa$-1)-th iteration $C_{i,k}^{\mathtt{D}}(\mathbf{X}_i,\alpha_i)$ in \eqref{eq:changeOP3:a} is lower bounded by
\begin{IEEEeqnarray}{rCl}
  \frac{\ln\bigl(1+ \gamma_{i,k}^{\mathtt{D}}(\mathbf{X}_i)  \bigr)}{\alpha_i} &\geq& \mathtt{A}_{i,k}^{(\kappa)} - \mathtt{B}_{i,k}^{(\kappa)}\frac{\varphi_{i,k}(\mathbf{X}_i)}{\bigr(\Re\{\mathbf{h}_{i,k}^H\mathbf{w}_{i,k}\}\bigl)^2}- \mathtt{C}_{i,k}^{(\kappa)}\alpha_i \qquad
	\label{eq:Rate1kappro}
 \end{IEEEeqnarray}
where
\begin{IEEEeqnarray}{rCl}
\mathtt{A}_{i,k}^{(\kappa)} &\triangleq& 2\frac{\ln\bigl(1+\gamma_{i,k}^{\mathtt{D}}(\mathbf{X}_i^{(\kappa)})\bigr)}{\alpha_i^{(\kappa)}} + \frac{\gamma_{i,k}^{\mathtt{D}}(\mathbf{X}_i^{(\kappa)})}{\alpha_i^{(\kappa)}\bigl(\gamma_{i,k}^{\mathtt{D}}(\mathbf{X}_i^{(\kappa)}) +1\bigr)},\nonumber\\
\mathtt{B}_{i,k}^{(\kappa)} &\triangleq& \frac{\bigl(\gamma_{i,k}^{\mathtt{D}}(\mathbf{X}_i^{(\kappa)})\bigr)^2}{\alpha_i^{(\kappa)}\bigl(\gamma_{i,k}^{\mathtt{D}}(\mathbf{X}_i^{(\kappa)})+1\bigr)},
\mathtt{C}_{i,k}^{(\kappa)}\triangleq \frac{\ln\bigl(1+\gamma_{i,k}^{\mathtt{D}}(\mathbf{X}_i^{(\kappa)})\bigr)}{(\alpha_i^{(\kappa)})^2}.\nonumber\qquad
\end{IEEEeqnarray}
We make use of the following inequality 
\begin{eqnarray}
\|\mathbf{x}\|^2 &\geq& 2\Re\{(\mathbf{x}^{(\kappa)})^H\mathbf{x}\} - \|\mathbf{x}^{(\kappa)}\|^2\quad\label{eq:inequad}
\end{eqnarray}
with $\forall \mathbf{x}\in\mathbb{C}^N, \mathbf{x}^{(\kappa)}\in\mathbb{C}^N,$
 due to the convexity of the function $\|\mathbf{x}\|^2$ to further expose the hidden convexity of the right-hand side (RHS) of \eqref{eq:Rate1kappro} as
\begin{IEEEeqnarray}{rCl}
  \frac{\ln\bigl(1+ \gamma_{i,k}^{\mathtt{D}}(\mathbf{X}_i)  \bigr)}{\alpha_i} &\geq& \mathtt{A}_{i,k}^{(\kappa)} - \mathtt{B}_{i,k}^{(\kappa)}\frac{\varphi_{i,k}(\mathbf{X}_i)}{\theta_{i,k}^{(\kappa)}(\mathbf{w}_{i,k})}
	- \mathtt{C}_{i,k}^{(\kappa)}\alpha_i \nonumber\\
	&:=&C_{i,k}^{\mathtt{D},(\kappa)}(\mathbf{X}_i,\alpha_i)
	\label{eq:Rate1ConvexAppro}
 \end{IEEEeqnarray}
 over the trust region
\begin{equation}\label{eq:R1ktrust}
2\Re\{\mathbf{h}_{i,k}^H\mathbf{w}_{i,k}\}-\Re\{\mathbf{h}_{i,k}^H\mathbf{w}_{i,k}^{(\kappa)}\} > 0,\ \forall(i,k)\in\mathcal{D}
\end{equation}
where 
\[\theta_{i,k}^{(\kappa)}(\mathbf{w}_{i,k})\triangleq \Re\{\mathbf{h}_{i,k}^H\mathbf{w}_{i,k}^{(\kappa)}\}\Bigr(2\Re\{\mathbf{h}_{i,k}^H\mathbf{w}_{i,k}\}-\Re\{\mathbf{h}_{i,k}^H\mathbf{w}_{i,k}^{(\kappa)}\}\Bigl).\]
 Note that $C_{i,k}^{\mathtt{D},(\kappa)}(\mathbf{X}_i,\alpha_i)$ is a lower bounding concave function of $C_{i,k}^{\mathtt{D}}(\mathbf{X}_i,\alpha_i)$, which also satisfies
\begin{equation}\label{eq:R1kAgree}
C_{i,k}^{\mathtt{D},(\kappa)}\bigl(\mathbf{X}_i^{(\kappa)},\alpha_i^{(\kappa)}\bigl) = \frac{1}{\alpha_i^{(\kappa)}}\ln\bigl(1 + \gamma_{i,k}^{\mathtt{D}}(\mathbf{X}_i^{(\kappa)})\bigr).
\end{equation}
As a result,  \eqref{eq:changeOP3:a} can be iteratively replaced by the following inequality:
\begin{IEEEeqnarray}{rCl}\label{eq:R1kConvex}
C_{i,k}^{\mathtt{D},(\kappa)}(\mathbf{X}_i,\alpha_i) \geq \eta + \Gamma_{i,k}^{\mathtt{D}},\ \forall(i,k)\in\mathcal{D}.
\end{IEEEeqnarray}

By defining  $\gamma_{i,\ell}^{\mathtt{U}}(\mathbf{X}_i)\triangleq \rho_{i,\ell}^2\mathbf{g}_{i,\ell}^H\boldsymbol{\Phi}_{i,\ell}(\mathbf{X}_i)^{-1}\mathbf{g}_{i,\ell}$, the left-hand side (LHS) of \eqref{eq:changeOP3:c} is lower bounded at the feasible point $\bigl(\mathbf{X}_i^{(\kappa)},\alpha_i^{(\kappa)}\bigr)$ as
\begin{IEEEeqnarray}{rCl}
   \frac{\ln\bigl(1 + \gamma_{i,\ell}^{\mathtt{U}}(\mathbf{X}_i)  \bigr)}{\alpha_i} &\geq& \tilde{\mathtt{A}}_{i,\ell}^{(\kappa)} + \tilde{\mathtt{B}}_{i,\ell}^{(\kappa)}\rho_{i,\ell} - \frac{\phi_{i,\ell}^{(\kappa)}\bigl(\mathbf{X}_i\bigr)}{\alpha_i^{(\kappa)}} - \tilde{\mathtt{C}}_{i,\ell}^{(\kappa)}\alpha_i \nonumber\\
	&:=& C_{i,\ell}^{\mathtt{U},(\kappa)}(\mathbf{X}_i,\alpha_i)
	\label{eq:Rate1qappro}
 \end{IEEEeqnarray}
where 
\begin{IEEEeqnarray}{rCl}
&&\tilde{\mathtt{A}}_{i,\ell}^{(\kappa)} \triangleq  \frac{2\ln\bigl(1+\gamma_{i,\ell}^{\mathtt{U}}(\mathbf{X}_i^{(\kappa)})\bigr) - \gamma_{i,\ell}^{\mathtt{U}}(\mathbf{X}_i^{(\kappa)})}{\alpha_i^{(\kappa)}},\ \nonumber\\
&&\tilde{\mathtt{B}}_{i,\ell}^{(\kappa)} \triangleq \frac{2\gamma_{i,\ell}^{\mathtt{U}}(\mathbf{X}_i^{(\kappa)})}{\rho_{i,\ell}^{(\kappa)}\alpha_i^{(\kappa)}},\   \tilde{\mathtt{C}}_{i,\ell}^{(\kappa)} \triangleq   \frac{\ln\bigl(1+\gamma_{i,\ell}^{\mathtt{U}}(\mathbf{X}_i^{(\kappa)})\bigr)}{\bigl(\alpha_i^{(\kappa)}\bigr)^2},                          \nonumber\\
&&\phi_{i,\ell}^{(\kappa)}\bigl(\mathbf{X}_i\bigr) \triangleq \sum_{j=\ell}^L\rho_{i,j}^2\mathbf{g}_{i,j}^H\boldsymbol{\Omega}_{i,\ell}^{(\kappa)}\mathbf{g}_{i,j}+\sigma^2\tr\bigl(\boldsymbol{\Omega}_{i,\ell}^{(\kappa)}\bigr)\;+ \nonumber\\
&&\sigma_{\mathtt{SI}}\tr\bigl(\mathbf{V}_{i}^H\mathbf{G}_{\mathtt{SI}}\boldsymbol{\Omega}_{i,\ell}^{(\kappa)}\mathbf{G}_{\mathtt{SI}}^H\mathbf{V}_{i}\bigr)+ \sigma_{\mathtt{SI}}\sum_{k=1}^K\mathbf{w}_{i,k}^H\mathbf{G}_{\mathtt{SI}}\boldsymbol{\Omega}_{i,\ell}^{(\kappa)}\mathbf{G}_{\mathtt{SI}}^H\mathbf{w}_{i,k}, \nonumber\\
 &&\boldsymbol{\Omega}_{i,\ell}^{(\kappa)} \triangleq \boldsymbol{\Phi}_{i,\ell}\bigl(\mathbf{X}_i^{(\kappa)}\bigr)^{-1} - \boldsymbol{\Phi}_{i,\ell-1}\bigl(\mathbf{X}_i^{(\kappa)}\bigr)^{-1}\succeq\mathbf{0}. \nonumber\qquad
\end{IEEEeqnarray}
It follows from \eqref{eq:Rate1qappro} that $C_{i,\ell}^{\mathtt{U},(\kappa)}(\mathbf{X}_i,\alpha_i)$ is a concave  function, which agrees with $C_{i,\ell}^{\mathtt{U}}(\mathbf{X}_i,\alpha_i)$ at the feasible point $\bigl(\mathbf{X}_i^{(\kappa)},\alpha_i^{(\kappa)}\bigr)$ as
\begin{IEEEeqnarray}{rCl}\label{eq:R1qAgree}
C_{i,\ell}^{\mathtt{U},(\kappa)}\bigl(\mathbf{X}_i^{(\kappa)},\alpha_i^{(\kappa)}\bigr) = 
\frac{1}{\alpha_i^{(\kappa)}}\ln\bigl(1 + \gamma_{i,\ell}^{\mathtt{U}}(\mathbf{X}_i^{(\kappa)}) \bigr).
\end{IEEEeqnarray}
Thus,  constraint \eqref{eq:changeOP3:c} can be iteratively 
replaced by
\begin{equation}\label{eq:R1qConvexappro}
C_{i,\ell}^{\mathtt{U},(\kappa)}\bigl(\mathbf{X}_i,\alpha_i\bigr) \geq \eta + \Gamma_{i,\ell}^{\mathtt{U}},\ \forall(i,\ell)\in\mathcal{U}.
\end{equation}

\textit{Approximation of  Constraints \eqref{eq:RateEvamChange}:} Notice that the LHSs of \eqref{eq:RateEvamChange} are neither convex nor concave with respect to (w.r.t.) $(\mathbf{X}_i,\alpha_i)$. For a given feasible point $x^{(\kappa)}$, the following inequality holds true:
\begin{equation}\label{ineupp}
\ln(1+x) \leq \mathtt{a}(x^{(\kappa)}) + \mathtt{b}(x^{(\kappa)})x,\;\forall x^{(\kappa)} \geq 0, x\geq 0
\end{equation}
where 
\begin{IEEEeqnarray}{rCl}
\mathtt{a}(x^{(\kappa)}) \triangleq \ln(1+x^{(\kappa)}) - \frac{x^{(\kappa)}}{1+x^{(\kappa)}},\
 \mathtt{b}(x^{(\kappa)}) \triangleq  \frac{1}{1+x^{(\kappa)}},\nonumber
	\end{IEEEeqnarray} 
which is a result of the concavity of the function $\ln(1+x)$. In the sequel, we develop an upper bound of the LHSs of \eqref{eq:RateEvamChange}. Let us consider \eqref{eq:RateEvamChange:a} first. For $\gamma_{m,i,k}^{\mathtt{ED},(\kappa)}\triangleq \|\mathbf{H}_{m}^H\mathbf{w}_{i,k}^{(\kappa)}\|^2/\psi_{m,i,k}(\mathbf{X}_i^{(\kappa)})$ at a feasible point $\mathbf{X}_i^{(\kappa)}$ and using \eqref{ineupp}, its LHS is upper bounded by
\begin{equation}\label{eq:RateEvamUpper:a}
\frac{\ln\Bigl(1 + \frac{\|\mathbf{H}_{m}^H\mathbf{w}_{i,k}\|^2}{\psi_{m,i,k}(\mathbf{X}_i)}\Bigr)}{\alpha_i} \leq \frac{\mathtt{a}(\gamma_{m,i,k}^{\mathtt{ED},(\kappa)})}{\alpha_i} + \mathtt{b}(\gamma_{m,i,k}^{\mathtt{ED},(\kappa)})\frac{\|\mathbf{H}_{m}^H\mathbf{w}_{i,k}\|^2}{\alpha_i\psi_{m,i,k}(\mathbf{X}_i)}.\qquad
\end{equation}
 We then introduce  new variables $\mu_{m,i,k} > 0, \forall m, i, k$, to decompose \eqref{eq:RateEvamChange:a} into a set of the following constraints:
\begin{IEEEeqnarray}{rCl}\label{eq:RateEvamUpper:a1}
&&C_{m,i,k}^{\mathtt{ED},(\kappa)}(\mathbf{X}_i,\alpha_i,\mu_{m,i,k}) \leq \Gamma_{i,k}^{\mathtt{D}}, \forall m,i,k,\IEEEyessubnumber\label{eq:RateEvamUpper:a1a}\qquad\\
&& \mu_{m,i,k}  \leq \alpha_i\psi_{m,i,k}(\mathbf{X}_i), \forall m,i,k\IEEEyessubnumber\label{eq:RateEvamUpper:a1b}
	\end{IEEEeqnarray}
where \eqref{eq:RateEvamUpper:a1a} is a convex constraint, and
\begin{IEEEeqnarray}{rCl}\label{eq:RateEvamUpper:a3}
C_{m,i,k}^{\mathtt{ED},(\kappa)}(\mathbf{X}_i,\alpha_i,\mu_{m,i,k}) &:=&\frac{\mathtt{a}(\gamma_{m,i,k}^{\mathtt{ED},(\kappa)})}{\alpha_i} \nonumber\\
 &+& \mathtt{b}(\gamma_{m,i,k}^{\mathtt{ED},(\kappa)}) \frac{\|\mathbf{H}_{m}^H\mathbf{w}_{i,k}\|^2}{\mu_{m,i,k}}.
\nonumber\end{IEEEeqnarray}
For the nonconvex constraint \eqref{eq:RateEvamUpper:a1b}, we can express it in more tractable form as
\begin{IEEEeqnarray}{rCl}\label{eq:RateEvamUpper:a4}
 \mathtt{F}_{m,i,k}(\mathbf{X}_i,\alpha_i,\mu_{m,i,k})  \leq    N_{e,m}\sigma^2
\end{IEEEeqnarray}
where 
\begin{IEEEeqnarray}{rCl}
\mathtt{F}_{m,i,k}(\mathbf{X}_i,\alpha_i,\mu_{m,i,k}) \triangleq \frac{\mu_{m,i,k}}{\alpha_i}-\Bigl[\psi_{m,i,k}(\mathbf{X}_i)-N_{e,m}\sigma^2\Bigr]. \nonumber
\end{IEEEeqnarray}
Let us define $\mathtt{F}_{m,i,k}^{(\kappa)}(\mathbf{X}_i,\alpha_i,\mu_{m,i,k})$ to be the convex approximation of $\mathtt{F}_{m,i,k}(\mathbf{X}_i,\alpha_i,\mu_{m,i,k})$ at $(\mathbf{X}_i^{(\kappa)},\alpha_i^{(\kappa)},\mu_{m,i,k}^{(\kappa)})$, of which the derivation is given in Appendix B. As a consequence, the convex approximation of  \eqref{eq:RateEvamUpper:a4} reads as
\begin{IEEEeqnarray}{rCl}\label{eq:RateEvamUpper:a5}
 \mathtt{F}_{m,i,k}^{(\kappa)}(\mathbf{X}_i,\alpha_i,\mu_{m,i,k}) \leq  N_{e,m}\sigma^2, \forall m,i,k. \quad\
\end{IEEEeqnarray}
Moreover, at the feasible point $\bigl(\mathbf{X}_i^{(\kappa)},\alpha_i^{(\kappa)},\mu_{m,i,k}^{(\kappa)}\bigr)$, it is also  true that
\begin{equation}\label{eq:RateEvamUpper:a6}
  C_{m,i,k}^{\mathtt{ED},(\kappa)}\bigl(\mathbf{X}_i^{(\kappa)},\alpha_i^{(\kappa)},\mu_{m,i,k}^{(\kappa)}\bigr) = \frac{\ln\bigl(1 + \gamma_{m,i,k}^{\mathtt{ED},(\kappa)}\bigr)}{\alpha_i^{(\kappa)}}
\end{equation}
since the inequality \eqref{eq:RateEvamUpper:a1b} must hold with equality at optimum.

Similarly, constraint  \eqref{eq:RateEvamChange:c} is iteratively approximated by
\begin{IEEEeqnarray}{rCl}\label{eq:RateEvamConvexAppro}
C_{m,i,\ell}^{\mathtt{EU},(\kappa)}(\mathbf{X}_i,\alpha_i,\tilde{\mu}_{m,i,\ell}) &\leq& \Gamma_{i,\ell}^{\mathtt{U}}, \forall m,i,\ell\
\end{IEEEeqnarray}
 with the additional convex   constraint
\begin{IEEEeqnarray}{rCl}\label{eq:RateEvamUpper:bcd}
 \mathtt{P}_{m,i,\ell}^{(\kappa)}(\mathbf{X}_i,\alpha_i,\tilde{\mu}_{m,i,\ell}) \leq  N_{e,m}\sigma^2, \forall m,i,\ell\
\end{IEEEeqnarray}
where $\tilde{\mu}_{m,i,\ell} > 0, \forall m,i,\ell$, are new variables, and
\begin{IEEEeqnarray}{rCl}\label{eq:RateEvamUpper:bcd5}
C_{m,i,\ell}^{\mathtt{EU},(\kappa)}(\mathbf{X}_i,\alpha_i,\tilde{\mu}_{m,i,\ell}) &:=&\frac{\mathtt{a}(\gamma_{m,i,\ell}^{\mathtt{EU},(\kappa)})}{\alpha_i} \nonumber\\ &+& \mathtt{b}(\gamma_{m,i,\ell}^{\mathtt{EU},(\kappa)})\frac{\rho_{i,\ell}^2\|\mathbf{g}_{m,i,\ell}^H\|^2}{\tilde{\mu}_{m,i,\ell}},
\nonumber\end{IEEEeqnarray}
for $\gamma_{m,i,\ell}^{\mathtt{EU},(\kappa)}\triangleq \bigl(\rho_{i,\ell}^{(\kappa)}\bigr)^2\|\mathbf{g}_{m,i,\ell}^H\|^2/\chi_{m,i,\ell}\bigl(\mathbf{X}_i^{(\kappa)}\bigr)$. The function $\mathtt{P}_{m,i,\ell}^{(\kappa)}(\mathbf{X}_i,\alpha_i,\tilde{\mu}_{m,i,\ell})$ in \eqref{eq:RateEvamUpper:bcd} is the convex approximation of 
\[\mathtt{P}_{m,i,\ell}(\mathbf{X}_i,\alpha_i,\tilde{\mu}_{m,i,\ell})\triangleq \frac{\tilde{\mu}_{m,i,\ell}}{\alpha_i}-\Bigl[\chi_{m,i,\ell}(\mathbf{X}_i)-N_{e,m}\sigma^2\Bigr],
\]
 of which the derivation is given in Appendix B.

\textit{Inner Approximation of Power Constraints \eqref{eq:OP4:c} and \eqref{eq:OP4:d}:}
Owing to the convexity of  functions $\|\mathbf{x}\|^2/t$,  we have \cite{Nguyen:TCOM:17}
\begin{eqnarray}
\frac{\|\mathbf{x}\|^2}{t} &\geq& \frac{2\Re\{(\mathbf{x}^{(\kappa)})^H\mathbf{x}\}}{t^{(\kappa)}} - \frac{\|\mathbf{x}^{(\kappa)}\|^2}{(t^{(\kappa)})^2}t\quad\label{B1}
\end{eqnarray}
with $\forall \mathbf{x}\in\mathbb{C}^N, \mathbf{x}^{(\kappa)}\in\mathbb{C}^N, t > 0, t^{(\kappa)} > 0$.
 We make use of the inequality \eqref{B1} to obtain  inner convex approximations for nonconvex constraints \eqref{eq:OP4:c} and \eqref{eq:OP4:d} as
\begin{IEEEeqnarray}{rCl}\label{eq:PowerAppro}
 \sum_{k=1}^K\|\mathbf{w}_{1,k}\|^2 + \|\mathbf{V}_1\|_\F^2 + \frac{1}{\alpha_2}\Bigl(\sum_{k=1}^{K}\|\mathbf{w}_{2,k}\|^2 + \|\mathbf{V}_2\|_\F^2\Bigr)\, +  \nonumber\\
\Bigl(\sum_{k=1}^K\|\mathbf{w}_{1,k}^{(\kappa)}\|^2+ \|\mathbf{V}_1^{(\kappa)}\|_\F^2\Bigr)\frac{\alpha_2}{(\alpha_2^{(\kappa)})^2}-\Bigr(\sum_{k=1}^K\Re\bigl\{(\mathbf{w}_{1,k}^{(\kappa)})^H\mathbf{w}_{1,k}\bigr\}  \nonumber\\
 +\;  \Re\left\{\tr\bigl((\mathbf{V}_1^{(\kappa)})^H\mathbf{V}_1\bigr)\right\}\Bigl)\frac{2}{\alpha_2^{(\kappa)}}  \leq  P_{bs}^{\max} ,   \IEEEyessubnumber\label{eq:OP4convex:c}\qquad\quad\\
 \rho_{1,\ell}^2 - \frac{2\rho_{1,\ell}^{(\kappa)}}{\alpha_2^{(\kappa)}}\rho_{1,\ell} + \frac{(\rho_{1,\ell}^{(\kappa)})^2}{(\alpha_2^{(\kappa)})^2}\alpha_2 \leq P_{1,\ell}^{\max},\;\forall \ell\in\mathcal{L}.\qquad\quad  \IEEEyessubnumber\label{eq:OP4convex:e}
\end{IEEEeqnarray}

Bearing all the above in mind, the following convex  program is solved at the $\kappa$-th iteration:
\begin{IEEEeqnarray}{rCl}\label{eq:OP5}
&&\underset{\mathbf{X},\eta,\boldsymbol{\Gamma},\boldsymbol{\alpha},\boldsymbol{\mu}}{\maxi}\quad \eta \IEEEyessubnumber\label{eq:OP5:a}\\
&&\st\; \eqref{eq:OP1:e}, \eqref{eq:changevariables:b},  \eqref{eq:OP4:e}, \eqref{eq:poscondi}, \eqref{eq:R1ktrust}, \eqref{eq:R1kConvex},  \nonumber \\
 &&\qquad  \eqref{eq:R1qConvexappro},  \eqref{eq:RateEvamUpper:a1a}, \eqref{eq:RateEvamUpper:a5}, \eqref{eq:RateEvamConvexAppro}, \eqref{eq:RateEvamUpper:bcd}, \eqref{eq:PowerAppro},\IEEEyessubnumber\label{eq:OP5:b}\\
&&\qquad  \mu_{m,i,k}>0,  \tilde{\mu}_{m,i,\ell} >0,\forall m,i,k,\ell  \qquad\IEEEyessubnumber\label{eq:OP5:c}
\end{IEEEeqnarray}
where $\boldsymbol{\mu}\triangleq\{\mu_{m,i,k}, \tilde{\mu}_{m,i,\ell}\}_{m\in\mathcal{M},i\in\mathcal{I},k\in\mathcal{K},\ell\in\mathcal{L}}$ represents the collection of the auxiliary variables.
A pseudo code-based path-following procedure  used to solve \eqref{eq:OP1} is given in Algorithm \ref{algo:knownCSI}. After solving \eqref{eq:OP5}, we update the involved variables until convergence, i.e., until the difference between two successive  values of the objective is less than a given threshold $\epsilon$.  

\begin{algorithm}[t]
\begin{algorithmic}[1]
\protect\caption{Proposed Path-Following  Algorithm to Solve SRM-EWCI  \eqref{eq:OP1}}
\label{algo:knownCSI}
\global\long\def\algorithmicrequire{\textbf{Initialization:}}
\REQUIRE  Set $\kappa:=0$ and solve \eqref{eq:OP5feasbile} to generate an initial feasible point $(\mathbf{X}^{(0)},\boldsymbol{\alpha}^{(0)},\boldsymbol{\mu}^{(0)})$.
\REPEAT
\STATE Solve \eqref{eq:OP5} to obtain the optimal solution ($\mathbf{X}^{\star},\eta^{\star},\boldsymbol{\Gamma}^{\star},\boldsymbol{\alpha}^{\star},\boldsymbol{\mu}^{\star}$).
\STATE Update $\mathbf{X}^{(\kappa+1)}:=\mathbf{X}^{\star},\boldsymbol{\alpha}^{(\kappa+1)}:=\boldsymbol{\alpha}^{\star},\boldsymbol{\mu}^{(\kappa+1)}:=\boldsymbol{\mu}^{\star}$.
\STATE Set $\kappa:=\kappa+1.$
\UNTIL Convergence\\
\end{algorithmic} \end{algorithm}

\textit{Practical Implementations:}
The first step of Algorithm \ref{algo:knownCSI}  requires a feasible point $(\mathbf{X}^{(0)},\boldsymbol{\alpha}^{(0)},\boldsymbol{\mu}^{(0)})$ of \eqref{eq:OP4} to  successfully initialize   the computational procedure, which is difficult to find in general. To circumvent this problem, initialized by any feasible $(\mathbf{X}^{(0)},\boldsymbol{\alpha}^{(0)})$ to the convex constraints \{\eqref{eq:OP1:e}, \eqref{eq:changevariables:b},  \eqref{eq:OP4:e}, \eqref{eq:poscondi}, \eqref{eq:R1ktrust}, \eqref{eq:R1kConvex}, \eqref{eq:R1qConvexappro}, \eqref{eq:PowerAppro}\}, the following convex program
\begin{IEEEeqnarray}{rCl}\label{eq:OP5feasbile}
\underset{\mathbf{X},\eta,\boldsymbol{\Gamma},\boldsymbol{\alpha}}{\maxi}&&\;\{\eta - \bar{\eta}_{\min}\} \IEEEyessubnumber\label{eq:OP5feasbile:a}\\
\st\;&& \eqref{eq:OP1:e}, \eqref{eq:changevariables:b},  \eqref{eq:OP4:e}, \eqref{eq:poscondi}, \eqref{eq:R1ktrust}, \eqref{eq:R1kConvex}, \eqref{eq:R1qConvexappro}, \eqref{eq:PowerAppro}\IEEEyessubnumber\label{eq:OP5feasbile:b}\quad\
\end{IEEEeqnarray}
without imposing Eves' constraints, is successively solved until reaching: $\{\eta - \bar{\eta}_{\min}\} \geq 0$. Herein, $\bar{\eta}_{\min} > 0$ is a given value to further improve the convergence speed of solving \eqref{eq:OP4}. The initial feasible $\boldsymbol{\mu}^{(0)}$ is then found by calculating: $\mu_{m,i,k}^{(0)}  = \alpha_i^{(0)}\psi_{m,i,k}(\mathbf{X}_i^{(0)})$ and $\tilde{\mu}_{m,i,\ell}^{(0)}  = \alpha_i^{(0)}\chi_{m,i,\ell}(\mathbf{X}_i^{(0)})$. We have numerically observed that  it usually takes about 2 iterations to generate an initial feasible point of \eqref{eq:OP4}.

\subsection{Convergence and Complexity Analysis}

For the sake of notational convenience, let us define the set of optimization variables $\Psi^{(\kappa)}$, the feasible set $\mathcal{V}^{(\kappa)}$, and the objective  $\eta^{(\kappa)}$ of the problem \eqref{eq:OP5} at the $\kappa$-th iteration by 
\begin{IEEEeqnarray}{rCl}
\Psi^{(\kappa)} &\triangleq& \left\{\mathbf{X},\eta,\boldsymbol{\Gamma},\boldsymbol{\alpha},\boldsymbol{\mu}\right\}, \nonumber \\
\mathcal{V}^{(\kappa)} &\triangleq& \left\{\Psi^{(\kappa)}\bigl|\st\ \eqref{eq:OP5:b}\ \text{and}\ \eqref{eq:OP5:c}\right\}, \nonumber\\
\eta^{(\kappa)} &\triangleq& \left\{\max \eta\bigl|\st\  \Psi^{(\kappa)}\in \mathcal{V}^{(\kappa)}\right\}, \nonumber 
\end{IEEEeqnarray}
respectively.

\begin{proposition}\label{pro:1}
The feasible set $\mathcal{V}^{(\kappa)}$ of the optimization problem \eqref{eq:OP5} is also feasible for \eqref{eq:OP4} (and hence the original problem \eqref{eq:OP1}).
\end{proposition}

\begin{proof}
See Appendix C.
\end{proof}
\noindent From Proposition \ref{pro:1}, it is ready to show the convergence of Algorithm \ref{algo:knownCSI}, which is stated as the following
proposition.
\begin{proposition}\label{pro:2}
Algorithm \ref{algo:knownCSI} yields a non-decreasing sequence of the objective, i.e. $\eta^{(\kappa+1)} \geq \eta^{(\kappa)}$. After finitely many iterations, Algorithm \ref{algo:knownCSI} converges to the Karush-Kuhn-Tucker (KKT) point of \eqref{eq:OP4} (and hence \eqref{eq:OP1}). 
\end{proposition}

\begin{proof}
See Appendix D.
\end{proof}

\textit{Complexity Analysis}: The computational complexity of solving \eqref{eq:OP5} in each iteration is only polynomial in the number of   constraints and optimization variables. To see this, the convex problem \eqref{eq:OP5} involves $(4M+6)(K+L)+2$ linear and quadratic constraints, and $y\triangleq 2N_{\mathtt{t}}^2+2(N_{\mathtt{t}}+M+1)K+2(M+2)L+3$ scalar
real variables. The complexity required to  solve \eqref{eq:OP5} is thus $\mathcal{O}\Bigr(\bigl((4M+6)(K+L)+2\bigl)^{2.5}\bigl(y^2+(4M+6)(K+L)+2\bigl)\Bigl)$.

\section{Secure Transmission Design with  SCSI on Eves}\label{sec:ECDI}
The perfect instantaneous CSI of the Eves may be difficult to obtain in some cases. For instance, Eves can move to another location to wiretap  confidential messages more effectively by protecting its visibility from  transmitters without exposing its CSI or they are  passive users. As a result, the CSI of Eves may change,  and thus,  FD-BS and UL users can only estimate Eves' channel statistics based on  their last known CSI. In this section,  we consider the case in which the FD-BS and UL users have perfect CSI of all  legitimate users as explained in Section \ref{sec:knownCSI}, but only the statistics of CSI  for Eves (i.e., the first-order and second-order statistics) \cite{AkgunTCOM17,WuTCOM17} as follows:
\begin{IEEEeqnarray}{rCl}
\bar{\mathbf{H}}_m &=& \mathbb{E}\left\{\mathbf{H}_m\mathbf{H}_m^H\right\}, \forall m\in\mathcal{M},\nonumber\\
\bar{g}_{m,i,\ell} &=&  \mathbb{E}\left\{\mathbf{g}_{m,i,\ell}\mathbf{g}_{m,i,\ell}^H\right\}, \forall m\in\mathcal{M}, (i,\ell)\in\mathcal{U}.
\label{eq:CDI}\end{IEEEeqnarray}

 Upon the above settings and similar to \eqref{eq:OP4}, the SRM problem \eqref{eq:OP1} with SCSI for Eves,  called SRM-SCSI, can be  re-expressed as
\begin{IEEEeqnarray}{rCl}\label{eq:OP6}
\underset{\mathbf{X},\eta,\boldsymbol{\Gamma},\boldsymbol{\alpha}}{\maxi}&&\quad \eta \IEEEyessubnumber\label{eq:OP6:a}\\
\st\;&&  \eqref{eq:OP1:e}, \eqref{eq:changevariables:b}, \eqref{eq:LegitimateUEs}, \eqref{eq:OP4:c},  \eqref{eq:OP4:d}, \eqref{eq:OP4:e},\IEEEyessubnumber\label{eq:OP6:b}\\
&&\underset{m\in\mathcal{M}}{\max}\, C_{m,i,k}^{\mathtt{ED}}(\mathbf{X}_i,\alpha_i)\leq \Gamma_{i,k}^{\mathtt{D}},   \forall(i,k)\in\mathcal{D},   \IEEEyessubnumber\label{eq:OP6:c}\\
&&\underset{m\in\mathcal{M}}{\max}\,C_{m,i,\ell}^{\mathtt{EU}}(\mathbf{X}_i,\alpha_i)\leq \Gamma_{i,\ell}^{\mathtt{U}}, \forall(i,\ell)\in\mathcal{U}.   \IEEEyessubnumber\label{eq:OP6:d}\qquad\
 \end{IEEEeqnarray}
We note that  problem \eqref{eq:OP6} is also nonconvex and more computationally difficult than solving \eqref{eq:OP4}. Fortunately,
it is clear that the convex approximations presented in Section \ref{sec:knownCSI}  are useful to obtain a solution to \eqref{eq:OP6}. To be specific, the nonconvex constraints in \eqref{eq:LegitimateUEs} were convexified by \eqref{eq:R1kConvex} and \eqref{eq:R1qConvexappro}, while \eqref{eq:OP4:c} and  \eqref{eq:OP4:d} were innerly approximated by the convex constraints \eqref{eq:OP4convex:c} and \eqref{eq:OP4convex:e}, respectively. The main obstacle to solving \eqref{eq:OP6} is to handle \eqref{eq:OP6:c} and \eqref{eq:OP6:d}. Besides, it may not be possible to design a completely secure solution because SCSI for Eves is unbounded and  random. Toward  a safe design \cite{Nguyen:TIFS:16,Ng-14-A},  we consider the replacement of  constraints \eqref{eq:OP6:c} and \eqref{eq:OP6:d} by their minimum outage requirement as
\begin{IEEEeqnarray}{rCl}\label{eq:OP7}
&&\underset{\mathbf{X},\eta,\boldsymbol{\Gamma},\boldsymbol{\alpha}}{\maxi}\quad \eta \IEEEyessubnumber\label{eq:OP7:a}\\
&&\st\;  \eqref{eq:OP1:e}, \eqref{eq:changevariables:b}, \eqref{eq:LegitimateUEs}, \eqref{eq:OP4:c},  \eqref{eq:OP4:d}, \eqref{eq:OP4:e},\IEEEyessubnumber\label{eq:OP7:b}\\
&&\Pro\Bigl(\underset{m\in\mathcal{M}}{\max} C_{m,i,k}^{\mathtt{ED}}(\mathbf{X}_i,\alpha_i)\leq \Gamma_{i,k}^{\mathtt{D}}\Bigl) \geq \epsilon_{i,k},   \forall(i,k)\in\mathcal{D},   \IEEEyessubnumber\label{eq:OP7:c}\qquad\\
&&\Pro\Bigl(\underset{m\in\mathcal{M}}{\max}C_{m,i,\ell}^{\mathtt{EU}}(\mathbf{X}_i,\alpha_i)\leq \Gamma_{i,\ell}^{\mathtt{U}}\Bigl)\geq \tilde{\epsilon}_{i,\ell}, \forall(i,\ell)\in\mathcal{U}.   \IEEEyessubnumber\label{eq:OP7:d}
 \end{IEEEeqnarray}
It is stated that the probabilities of  Eves' maximum allowable rates   targeted for DL user ($i,k$) and UL user $(i,\ell)$ should be   greater than or equal to certain constants  $\epsilon_{i,k}$ and $\tilde{\epsilon}_{i,\ell}$, respectively. In fact, $\epsilon_{i,k}$ and $\tilde{\epsilon}_{i,\ell}$ are required to be close to 1 to provide secure communications. To evaluate \eqref{eq:OP7:c} and \eqref{eq:OP7:d}, we first introduce the following lemma.
\begin{lemma}\label{lemma:1}
Assuming all Eves have the same channel properties and they are independent,  \eqref{eq:OP7:c} and \eqref{eq:OP7:d}  are  converted  to the following constraints:
\begin{IEEEeqnarray}{rCl}
\frac{\mathbf{w}_{i,k}^H\bar{\mathbf{H}}_{m}\mathbf{w}_{i,k}}{e^{\alpha_i\Gamma_{i,k}^{\mathtt{D}}}-1}  \leq \bar{\psi}_{m,i,k}(\mathbf{X}_i) + \bigl(1-\epsilon_{i,k}^{1/M}\bigr)N_{e,m}\sigma^2,\nonumber\quad  \\
 \forall m\in\mathcal{M}, (i,k)\in\mathcal{D}, \label{eq:OP7:c1}\quad
 \end{IEEEeqnarray}
and 
\begin{IEEEeqnarray}{rCl}
\frac{\rho_{i,\ell}^2\bar{g}_{m,i,\ell} }{e^{\alpha_i\Gamma_{i,\ell}^{\mathtt{U}}}-1} \leq \bar{\chi}_{m,i,\ell}(\mathbf{X}_i) +\bigl(1-\tilde{\epsilon}_{i,\ell}^{1/M}\bigl) N_{e,m}\sigma^2,\nonumber\\
 \forall m\in\mathcal{M}, (i,\ell)\in\mathcal{U}\label{eq:OP7:d1}\quad
 \end{IEEEeqnarray}
 where 
\begin{IEEEeqnarray}{rCl}
\bar{\psi}_{m,i,k}(\mathbf{X}_i) &\triangleq&\mathtt{\bar{Q}}_{m,i}(\mathbf{X}_i)-\mathbf{w}_{i,k}^H\bar{\mathbf{H}}_{m}\mathbf{w}_{i,k},\nonumber\\
\bar{\chi}_{m,i,\ell}(\mathbf{X}_i) &\triangleq& \mathtt{\bar{Q}}_{m,i}(\mathbf{X}_i)-\rho_{i,\ell}^2\bar{g}_{m,i,\ell},\nonumber\\
 \mathtt{\bar{Q}}_{m,i}(\mathbf{X}_i)&\triangleq&\sum\nolimits_{k'=1}^K\mathbf{w}_{i,k'}^H\bar{\mathbf{H}}_{m}\mathbf{w}_{i,k'} +\tr\bigl(\mathbf{V}_{i}^H\bar{\mathbf{H}}_{m}\mathbf{V}_{i}\bigr)\nonumber\\
 &+&\; \sum\nolimits_{\ell^{'}=1}^L\rho_{i,\ell^{'}}^2\bar{g}_{m,i,\ell^{'}} \nonumber.
 \end{IEEEeqnarray}
\end{lemma}
\begin{proof}
See Appendix E.
\end{proof}

We note that  constraint \eqref{eq:OP7:c1} is still nonconvex but can be further shaped to take the equivalent  form:
\begin{IEEEeqnarray}{rCl}\label{eq:OP7:c3}
&&\frac{\mathbf{w}_{i,k}^H\bar{\mathbf{H}}_{m}\mathbf{w}_{i,k}}{\beta_{i,k}^{\mathtt{D}}}  \leq \bar{\psi}_{m,i,k}(\mathbf{X}_i) + \bigl(1-\epsilon_{i,k}^{1/M}\bigr)N_{e,m}\sigma^2, \IEEEyessubnumber\label{eq:OP7:c3a:nonconvex}\qquad\\
&& \beta_{i,k}^{\mathtt{D}} \leq e^{\alpha_i\Gamma_{i,k}^{\mathtt{D}}}-1 \Leftrightarrow \frac{\ln(1 + \beta_{i,k}^{\mathtt{D}})}{\alpha_i} \leq \Gamma_{i,k}^{\mathtt{D}} \IEEEyessubnumber\label{eq:OP7:c3b}\end{IEEEeqnarray}
where $\beta_{i,k}^{\mathtt{D}} > 0, \forall(i,k)\in\mathcal{D}$ are new variables. For \eqref{eq:OP7:c3a:nonconvex}, its LHS is a quadratic-over-affine function (which is convex) and the first term of the RHS is a quadratic function. Then, we iteratively replace \eqref{eq:OP7:c3a:nonconvex} by
\begin{IEEEeqnarray}{rCl}\label{eq:OP7:c31}
\frac{\mathbf{w}_{i,k}^H\bar{\mathbf{H}}_{m}\mathbf{w}_{i,k}}{\beta_{i,k}^{\mathtt{D}}}  \leq \bar{\psi}_{m,i,k}^{(\kappa)}(\mathbf{X}_i) + \bigl(1-\epsilon_{i,k}^{1/M}\bigr)N_{e,m}\sigma^2 \IEEEyessubnumber\label{eq:OP7:c3a},\nonumber\quad  \\
 \forall m\in\mathcal{M}, (i,k)\in\mathcal{D}\quad
\end{IEEEeqnarray}
where $\bar{\psi}_{m,i,k}^{(\kappa)}(\mathbf{X}_i)$ is the inner approximation of $\bar{\psi}_{m,i,k}(\mathbf{X}_i)$  using \eqref{eq:inequad} as
\begin{IEEEeqnarray}{rCl}\bar{\psi}_{m,i,k}^{(\kappa)}(\mathbf{X}_i) = \mathtt{\bar{Q}}_{m,i}^{(\kappa)}(\mathbf{X}_i) - 2\Re\bigl\{\bigl(\mathbf{w}_{i,k}^{(\kappa)}\bigr)^H\bar{\mathbf{H}}_{m}\mathbf{w}_{i,k}\bigl\} \nonumber\\
+\; (\mathbf{w}_{i,k}^{(\kappa)})^H \bar{\mathbf{H}}_{m}\mathbf{w}_{i,k}^{(\kappa)}, \nonumber
\end{IEEEeqnarray}
and $\mathtt{\bar{Q}}_{m,i}^{(\kappa)}(\mathbf{X}_i)$ is obtained by taking the expectation operations on each individual random terms  of $\mathtt{Q}_{m,i}^{(\kappa)}(\mathbf{X}_i)$  given in \eqref{B4} w.r.t. \eqref{eq:CDI}.
 By using \eqref{ineupp},  \eqref{eq:OP7:c3b} holds that
\begin{IEEEeqnarray}{rCl}\label{eq:OP7:c3b1}
\frac{\mathtt{a}(\beta_{i,k}^{\mathtt{D},{(\kappa)}})}{\alpha_i} + \mathtt{b}(\beta_{i,k}^{\mathtt{D},{(\kappa)}})\frac{\beta_{i,k}^{\mathtt{D}}}{\alpha_i} \leq \Gamma_{i,k}^{\mathtt{D}}.
\end{IEEEeqnarray}
For $\mathcal{W}(\beta_{i,k}^{\mathtt{D}},\alpha_i)\triangleq\beta_{i,k}^{\mathtt{D}}/\alpha_i$ and applying \eqref{B3} yields
\begin{IEEEeqnarray}{rCl}\label{eq:approW}
\mathcal{W}(\beta_{i,k}^{\mathtt{D}},\alpha_i) &\leq& \frac{1}{2}\Biggl(\frac{(\beta_{i,k}^{\mathtt{D}})^2}{\beta_{i,k}^{\mathtt{D},(\kappa)}}\frac{1}{\alpha_i^{(\kappa)}} + \beta_{i,k}^{\mathtt{D},(\kappa)}\frac{1}{2\alpha_i-\alpha_i^{(\kappa)}}\Biggr) \nonumber\\
&:=& \mathcal{W}^{(\kappa)}(\beta_{i,k}^{\mathtt{D}},\alpha_i).
\end{IEEEeqnarray}
As a result, the following inequality holds
\begin{IEEEeqnarray}{rCl}\label{eq:OP7:c3c1}
 \frac{\mathtt{a}(\beta_{i,k}^{\mathtt{D},{(\kappa)}})}{\alpha_i} + \mathtt{b}(\beta_{i,k}^{\mathtt{D},{(\kappa)}})\mathcal{W}^{(\kappa)}(\beta_{i,k}^{\mathtt{D}},\alpha_i) \leq \Gamma_{i,k}^{\mathtt{D}},\; \forall (i,k)\in\mathcal{D}\qquad
\end{IEEEeqnarray}
which is the convex approximation of \eqref{eq:OP7:c3b1}.

By following  steps \eqref{eq:OP7:c3}-\eqref{eq:OP7:c3c1}, we
 equivalently decompose \eqref{eq:OP7:d1} into  a set of the following convex constraints:
\begin{IEEEeqnarray}{rCl}\label{eq:OP7:d3}
&&\frac{\rho_{i,\ell}^2\bar{g}_{m,i,\ell}}{\beta_{i,\ell}^{\mathtt{U}}}  \leq \bar{\chi}^{(\kappa)}_{m,i,\ell}(\mathbf{X}_i) + \bigl(1-\tilde{\epsilon}_{i,\ell}^{1/M}\bigr)N_{e,m}\sigma^2,\nonumber\\  
&&\qquad\qquad\qquad\qquad\qquad\qquad\quad  \forall m\in\mathcal{M}, (i,\ell)\in\mathcal{U}, \IEEEyessubnumber\label{eq:OP7:d3a}\quad\\
&& \frac{\mathtt{a}(\beta_{i,\ell}^{\mathtt{U},{(\kappa)}})}{\alpha_i} + \mathtt{b}(\beta_{i,\ell}^{\mathtt{U},{(\kappa)}})\mathcal{W}^{(\kappa)}(\beta_{i,\ell}^{\mathtt{U}},\alpha_i) \leq \Gamma_{i,\ell}^{\mathtt{U}}, \forall (i,\ell)\in\mathcal{U}\qquad\IEEEyessubnumber\label{eq:OP7:d3c}\  \end{IEEEeqnarray}
where $\beta_{i,\ell}^{\mathtt{U}}>0, \forall(i,\ell)\in\mathcal{U}$ are new variables, and $\bar{\chi}^{(\kappa)}_{m,i,\ell}(\mathbf{X}_i)$ is the inner approximation  of $\bar{\chi}_{m,i,\ell}(\mathbf{X}_i)$ by using \eqref{eq:inequad}, which is given as
\begin{IEEEeqnarray}{rCl}\bar{\chi}^{(\kappa)}_{m,i,\ell}(\mathbf{X}_i) = \mathtt{\bar{Q}}_{m,i}^{(\kappa)}(\mathbf{X}_i) - \left(2\rho_{i,\ell}^{(\kappa)}\rho_{i,\ell} -  \bigl(\rho_{i,\ell}^{(\kappa)}\bigr)^2\right)\bar{g}_{m,i,\ell}. \nonumber
\end{IEEEeqnarray}

In summary,  the following convex program, which is an inner approximation of \eqref{eq:OP7}, is solved at the $\kappa$-th iteration:
\begin{IEEEeqnarray}{rCl}\label{eq:OP8}
\underset{\mathbf{X},\eta,\boldsymbol{\Gamma},\boldsymbol{\alpha},\boldsymbol{\beta}}{\maxi}&&\quad \eta \IEEEyessubnumber\label{eq:OP8:a}\\
\st\;&& \eqref{eq:OP1:e}, \eqref{eq:changevariables:b},  \eqref{eq:OP4:e}, \eqref{eq:poscondi}, \eqref{eq:R1ktrust}, \eqref{eq:R1kConvex},  \nonumber\\ 
  && \eqref{eq:R1qConvexappro},\eqref{eq:PowerAppro}, \eqref{eq:OP7:c31}, \eqref{eq:OP7:c3c1}, \eqref{eq:OP7:d3},\IEEEyessubnumber\label{eq:OP8:b}\\
&& \beta_{i,k}^{\mathtt{D}} > 0,  \beta_{i,\ell}^{\mathtt{U}} > 0,\  \forall(i,k)\in\mathcal{D}, (i,\ell)\in\mathcal{U}  \IEEEyessubnumber\label{eq:OP8:c}\qquad
\end{IEEEeqnarray}
where $\boldsymbol{\beta}\triangleq\{\beta_{i,k}^{\mathtt{D}},  \beta_{i,\ell}^{\mathtt{U}}\}_{i\in\mathcal{I},k\in\mathcal{K},\ell\in\mathcal{L}}$. The proposed Algorithm \ref{algo:unknownCSI} outlines the steps to solve \eqref{eq:OP7}. Specifically, we solve the convex program \eqref{eq:OP8}   in step 2 to obtain the optimal solution denoted by ($\mathbf{X}^{\star},\eta^{\star},\boldsymbol{\Gamma}^{\star},\boldsymbol{\alpha}^{\star}, \boldsymbol{\beta}^{\star}$) and update ($\mathbf{X}^{\star},\boldsymbol{\alpha}^{\star}, \boldsymbol{\beta}^{\star}$) as the feasible point for  the next iteration in step 3 of Algorithm \ref{algo:unknownCSI} until convergence. An initial feasible point to solve \eqref{eq:OP7} is easily found using the same method reported in Section \ref{sec:knownCSI}.

\begin{algorithm}[t]
\begin{algorithmic}[1]
\protect\caption{Proposed Path-Following  Algorithm to Solve SRM-SCSI  \eqref{eq:OP7} }
\label{algo:unknownCSI}
\global\long\def\algorithmicrequire{\textbf{Initialization:}}
\REQUIRE  Set $\kappa:=0$ and generate an initial feasible point $(\mathbf{X}^{(0)},\boldsymbol{\alpha}^{(0)},\boldsymbol{\beta}^{(0)})$.
\REPEAT
\STATE Solve \eqref{eq:OP8} to obtain the optimal solution ($\mathbf{X}^{\star},\eta^{\star},\boldsymbol{\Gamma}^{\star},\boldsymbol{\alpha}^{\star}, \boldsymbol{\beta}^{\star}$).
\STATE Update $\mathbf{X}^{(\kappa+1)}:=\mathbf{X}^{\star},\boldsymbol{\alpha}^{(\kappa+1)}:=\boldsymbol{\alpha}^{\star},\boldsymbol{\beta}^{(\kappa+1)}:=\boldsymbol{\beta}^{\star}$.
\STATE Set $\kappa:=\kappa+1.$
\UNTIL Convergence\\
\end{algorithmic} \end{algorithm}

\begin{remark}
An important remark here is that the feasible point for \eqref{eq:OP8} is also feasible for  problem  \eqref{eq:OP7}  but not vice versa because the results  in  \eqref{eq:OP7:c1} and \eqref{eq:OP7:d1} are obtained by using the Markov upper bound on the  outage probabilities for DL and UL users   (see Appendix E). Thus, the performance under the proposed method  serves as a lower bound on the SR performance, which  also coincides with the conclusion in \cite{Ng-14-A}.
\end{remark}

\section{Secure Transmission Design under Worst-case Scenario}\label{sec:MMSEEavs}
By introducing AN into a network, the signal strength distribution is essentially affected. The Eves can learn that  additional interference is caused by AN, i.e., conduct a characterization of the average power of the received signals and classify their shape (whether a poor wiretap link quality is
due to jamming) \cite{Xu:2005:FLD}. Therefore, each Eve can employ a more advanced linear decoder to combat such additional interference, i.e., the MMSE decoder for simplicity.  In addition, we assume that targeted message at an Eve is decoded after canceling all multiuser interference caused by DL  and UL users \cite{SunTWC16}. 
As a consequence, the information rates at  the $m$-th Eve in \eqref{eq:RateEvam}  can be reformulated as
\begin{IEEEeqnarray}{rCl}\label{eq:RateEvamRevise}
\hat{C}_{m,i,k}^{\mathtt{ED}}(\mathbf{X}_i,\tau_i) &=& \tau_i\ln\bigl(1 + \mathbf{w}_{i,k}^H\mathbf{H}_{m}\boldsymbol{\Xi}_{m,i}^{-1}\mathbf{H}_{m}^H\mathbf{w}_{i,k}\bigr),\IEEEyessubnumber\quad\\
\hat{C}_{m,i,\ell}^{\mathtt{EU}}(\mathbf{X}_i,\tau_i) &=& \tau_i\ln\bigl(1 + \rho_{i,\ell}^2\mathbf{g}_{m,i,\ell}\boldsymbol{\Xi}_{m,i}^{-1}\mathbf{g}_{m,i,\ell}^H\bigr)\IEEEyessubnumber
\end{IEEEeqnarray}
where
$\boldsymbol{\Xi}_{m,i} \triangleq\mathbf{H}_{m}^H\mathbf{V}_{i}\mathbf{V}_{i}^H\mathbf{H}_{m} + \sigma^2\mathbf{I}_{N_{e,m}}$, and for notational simplicity we dropped the arguments of  $\boldsymbol{\Xi}_{m,i}$. One can see that the AN signals contribute very much to neutralizing a higher
potential for information leakage to the Eves.

The above assumptions lead to a worst-case scenario  to provide secure communications of the SRM problem \eqref{eq:OP1}, called SRM-WCS, which can be mathematically formulated by
\begin{IEEEeqnarray}{rCl}\label{eq:OP9}
\underset{\mathbf{X},\boldsymbol{\tau}}{\maxi}&&\;\underset{\substack{(i,k)\in\mathcal{D}\\ (i,\ell)\in\mathcal{U}}}{\mini}\;\left\{\hat{R}_{i,k}^{\mathtt{D}}(\mathbf{X}_i,\tau_i),\hat{R}_{i,\ell}^{\mathtt{U}}(\mathbf{X}_i,\tau_i)\right\} \IEEEyessubnumber\label{eq:OP9:a}\quad\\
\st\ &&  \eqref{eq:OP1:b}-\eqref{eq:OP1:f}  \IEEEyessubnumber\label{eq:OP9:b}
\end{IEEEeqnarray}
where 
\[\hat{R}_{i,k}^{\mathtt{D}}(\mathbf{X}_i,\tau_i) \triangleq C_{i,k}^{\mathtt{D}}(\mathbf{X}_i,\tau_i) - \underset{m\in\mathcal{M}}{\max}\,\hat{C}_{m,i,k}^{\mathtt{ED}}(\mathbf{X}_i,\tau_i),\] 
\[\hat{R}_{i,\ell}^{\mathtt{U}}(\mathbf{X}_i,\tau_i) \triangleq C_{i,\ell}^{\mathtt{U}}(\mathbf{X}_i,\tau_i) - \underset{m\in\mathcal{M}}{\max}\,\hat{C}_{m,i,\ell}^{\mathtt{EU}}(\mathbf{X}_i,\tau_i).\]
 By exploiting the developments in Section \ref{sec:knownCSI},  problem \eqref{eq:OP9} can be transformed to
\begin{IEEEeqnarray}{rCl}\label{eq:OP10}
\underset{\mathbf{X},\eta,\boldsymbol{\Gamma},\boldsymbol{\alpha}}{\maxi}&&\quad \eta \IEEEyessubnumber\label{eq:OP10:a}\\
\st\;&& \eqref{eq:OP1:e}, \eqref{eq:changevariables:b},  \eqref{eq:OP4:e}, \eqref{eq:poscondi}, \eqref{eq:R1ktrust}, \eqref{eq:R1kConvex},    \eqref{eq:R1qConvexappro},   \eqref{eq:PowerAppro},\IEEEyessubnumber\label{eq:OP10:b}\\
&& \hat{C}_{m,i,k}^{\mathtt{ED}}(\mathbf{X}_i,\alpha_i)\leq \Gamma_{i,k}^{\mathtt{D}},  \forall m\in\mathcal{M}, (i,k)\in\mathcal{D},   \IEEEyessubnumber\label{eq:OP10:c}\qquad \\
&&\hat{C}_{m,i,\ell}^{\mathtt{EU}}(\mathbf{X}_i,\alpha_i)\leq \Gamma_{i,\ell}^{\mathtt{U}}, \forall m\in\mathcal{M}, (i,\ell)\in\mathcal{U}   \IEEEyessubnumber\label{eq:OP10:d}
\end{IEEEeqnarray}
where 
$\hat{C}_{m,i,k}^{\mathtt{ED}}(\mathbf{X}_i,\alpha_i) = \frac{1}{\alpha_i}\ln\bigl(1 + \mathbf{w}_{i,k}^H\mathbf{H}_{m}\boldsymbol{\Xi}_{m,i}^{-1}\mathbf{H}_{m}^H\mathbf{w}_{i,k}\bigr)$ and
$\hat{C}_{m,i,\ell}^{\mathtt{EU}}(\mathbf{X}_i,\alpha_i) = \frac{1}{\alpha_i}\ln\bigl(1 + \rho_{i,\ell}^2\mathbf{g}_{m,i,\ell}\boldsymbol{\Xi}_{m,i}^{-1}\mathbf{g}_{m,i,\ell}^H\bigr)$. We note that finding the tight convex approximations for   \eqref{eq:OP10:c} and \eqref{eq:OP10:d} are more challenging  than  those of \eqref{eq:RateEvamChange:a} and \eqref{eq:RateEvamChange:c}.

 We first introduce new variables $t_{m,i,k} > 0,\forall m,i,k$, to  equivalently express \eqref{eq:OP10:c} as
\begin{IEEEeqnarray}{rCl}\label{eq:RateEvamUpper:OP10c}
\frac{\ln\bigl(1 + t_{m,i,k}\bigr)}{\alpha_i}  &\leq&  \Gamma_{i,k}^{\mathtt{D}},\IEEEyessubnumber\label{eq:RateEvamUpper:OP10ca1}         \\ 
\mathbf{w}_{i,k}^H\mathbf{H}_{m}\boldsymbol{\Xi}_{m,i}^{-1}\mathbf{H}_{m}^H\mathbf{w}_{i,k} &\leq& t_{m,i,k}.\IEEEyessubnumber\label{eq:RateEvamUpper:OP10ca2}\qquad
\end{IEEEeqnarray}
Constraint \eqref{eq:RateEvamUpper:OP10ca1} follows from \eqref{eq:OP7:c3c1} that
\begin{IEEEeqnarray}{rCl}\label{eq:RateEvamUpper:OP10ca}
\frac{\mathtt{a}(t_{m,i,k}^{(\kappa)})}{\alpha_i} + \mathtt{b}(t_{m,i,k}^{(\kappa)})\mathcal{W}^{(\kappa)}(t_{m,i,k},\alpha_i) \leq \Gamma_{i,k}^{\mathtt{D}},\ \forall m,i,k.\quad
\end{IEEEeqnarray}
Although \eqref{eq:RateEvamUpper:OP10ca2} is still a nonconvex constraint, its tractable form can be obtained by
\begin{equation}\label{eq:10a:6}
\begin{bmatrix}
 t_{m,i,k}       & \mathbf{w}_{i,k}^H\mathbf{H}_{m}\\
   \mathbf{H}_{m}^H\mathbf{w}_{i,k}       & \boldsymbol{\Xi}_{m,i}
\end{bmatrix}\succeq \mathbf{0}
\end{equation}
which is a result of applying Schur complement \cite{Stephen}. Note that \eqref{eq:10a:6} is convex  if and only if this set  can be written as a linear matrix inequality (LMI). Let us define $\boldsymbol{\Xi}_{m,i} \triangleq \boldsymbol{\Sigma}_{m,i} + \sigma^2\mathbf{I}_{N_{e,m}}$, where $\boldsymbol{\Sigma}_{m,i} = \mathbf{H}_{m}^H\mathbf{V}_{i}\mathbf{V}_{i}^H\mathbf{H}_{m}$. The function  $\boldsymbol{\Sigma}_{m,i}$  can be  linearized  by  the linear mapping	 \cite{Nguyen:JSAC:17}:
\begin{IEEEeqnarray}{rCl}	\boldsymbol{\Sigma}_{m,i}^{(\kappa)} \triangleq
	\mathbf{H}_m^H\Bigl(\mathbf{V}_{i}(\mathbf{V}_{i}^{(\kappa)})^H + \mathbf{V}_{i}^{(\kappa)}\mathbf{V}_{i}^H 
	- \mathbf{V}_{i}^{(\kappa)}(\mathbf{V}_{i}^{(\kappa)})^H\Bigr)\mathbf{H}_m.\nonumber
\end{IEEEeqnarray}	
Thus, the LMI for \eqref{eq:10a:6} is given by
\begin{equation}\label{eq:10a:6a}
\begin{bmatrix}
  t_{m,i,k}      & \mathbf{w}_{i,k}^H\mathbf{H}_{m}\\
   \mathbf{H}_{m}^H\mathbf{w}_{i,k}       & \boldsymbol{\Sigma}_{m,i}^{(\kappa)} + \sigma^2\mathbf{I}_{N_{e,m}}
\end{bmatrix}\succeq \mathbf{0}, \forall m,i,k
\end{equation}
over the trust region 
\begin{IEEEeqnarray}{rCl}\label{eq:RateEvamUpper:OP10trust}
  \boldsymbol{\Sigma}_{m,i}^{(\kappa)} \succeq \mathbf{0},\quad \forall m\in\mathcal{M},i\in\mathcal{I}.
	\end{IEEEeqnarray}
	
Analogously,  applying  similar steps from \eqref{eq:RateEvamUpper:OP10c}-\eqref{eq:10a:6a} for \eqref{eq:OP10:d} yields
\begin{IEEEeqnarray}{rCl}\label{eq:RateEvamUpper:OP10d1}
\frac{\mathtt{a}(\tilde{t}_{m,i,\ell}^{(\kappa)})}{\alpha_i} + \mathtt{b}(\tilde{t}_{m,i,\ell}^{(\kappa)})\mathcal{W}^{(\kappa)}(\tilde{t}_{m,i,\ell},\alpha_i) \leq \Gamma_{i,\ell}^{\mathtt{U}}, \forall m,i,\ell,&&\IEEEyessubnumber\IEEEyessubnumber\qquad\\
\begin{bmatrix}
\tilde{t}_{m,i,\ell}      & \rho_{i,\ell}\mathbf{g}_{m,i,\ell}\\
   \rho_{i,\ell}\mathbf{g}_{m,i,\ell}^H     & \boldsymbol{\Sigma}_{m,i}^{(\kappa)} + \sigma^2\mathbf{I}_{N_{e,m}}
\end{bmatrix}\succeq \mathbf{0},
 \forall m,i,\ell&&\IEEEyessubnumber
\end{IEEEeqnarray}
where $\tilde{t}_{m,i,\ell} > 0$  are new variables.

In Algorithm \ref{algo:worstcase}, we propose a path-following algorithm to solve the SRM-WCS problem \eqref{eq:OP9}. At the $\kappa$-th iteration, it solves the following convex program: 
\begin{IEEEeqnarray}{rCl}\label{eq:OP11}
&&\underset{\mathbf{X},\eta,\boldsymbol{\Gamma},\boldsymbol{\alpha},\boldsymbol{t}}{\maxi}\quad \eta \IEEEyessubnumber\label{eq:OP11:a}\\
&&\st\; \eqref{eq:OP1:e}, \eqref{eq:changevariables:b},  \eqref{eq:OP4:e}, \eqref{eq:poscondi}, \eqref{eq:R1ktrust},   \eqref{eq:R1kConvex},  \nonumber\\
 &&\qquad \eqref{eq:R1qConvexappro},\eqref{eq:PowerAppro},  \eqref{eq:RateEvamUpper:OP10ca},     \eqref{eq:10a:6a}, \eqref{eq:RateEvamUpper:OP10trust}, \eqref{eq:RateEvamUpper:OP10d1},  \IEEEyessubnumber\label{eq:OP11:b}\\
&&\qquad t_{m,i,k}>0,  \tilde{t}_{m,i,\ell} >0,\ \forall m, i,k,\ell \qquad \IEEEyessubnumber\label{eq:OP11:c}
\end{IEEEeqnarray}
where $\boldsymbol{t}\triangleq\{t_{m,i,k}, \tilde{t}_{m,i,\ell}\}_{m\in\mathcal{M},i\in\mathcal{I},k\in\mathcal{K},\ell\in\mathcal{L}}$, to generate a sequence $\{\mathbf{X}^{(\kappa+1)},\boldsymbol{\alpha}^{(\kappa+1)},\boldsymbol{t}^{(\kappa+1)}\}$ (and hence $\{\mathbf{X}^{(\kappa+1)},\boldsymbol{\tau}^{(\kappa+1)}\}$) of
improved points of \eqref{eq:OP9}, which also converges to a KKT
point.
\begin{algorithm}[t]
\begin{algorithmic}[1]
\protect\caption{Proposed Path-Following  Algorithm to Solve SRM-WCS  \eqref{eq:OP9}}
\label{algo:worstcase}
\global\long\def\algorithmicrequire{\textbf{Initialization:}}
\REQUIRE  Set $\kappa:=0$ and generate an initial feasible point $(\mathbf{X}^{(0)},\boldsymbol{\alpha}^{(0)},\boldsymbol{t}^{(0)})$.
\REPEAT
\STATE Solve \eqref{eq:OP11} to obtain the optimal solution ($\mathbf{X}^{\star},\eta^{\star},\boldsymbol{\Gamma}^{\star},\boldsymbol{\alpha}^{\star},\boldsymbol{t}^{\star}$).
\STATE Update $\mathbf{X}^{(\kappa+1)}:=\mathbf{X}^{\star},\boldsymbol{\alpha}^{(\kappa+1)}:=\boldsymbol{\alpha}^{\star},\boldsymbol{t}^{(\kappa+1)}:=\boldsymbol{t}^{\star}$.
\STATE Set $\kappa:=\kappa+1.$
\UNTIL Convergence\\
\end{algorithmic} \end{algorithm}

\section{Numerical Results}\label{NumericalResults}
 
\begin{table}[t]
\caption{Simulation Parameters}
	\label{parameter}
	\centering
	{\setlength{\tabcolsep}{0.2em}
\setlength{\extrarowheight}{0.1em}
		\begin{tabular}{l|l}
		\hline
				Parameter & Value \\
		\hline\hline
		    Carrier center  frequency/ System bandwidth                            & 2 GHz/ 10 MHz \\
				Distance between the FD-BS and  nearest user  & $\geq$ 10 m\\
				Noise power spectral density at the receivers & -174 dBm/Hz \\
				PL model for LOS communications,  $\mathrm{PL}_{\mathtt{LOS}}$ & 103.8 + 20.9$\log_{10}(d)$ dB\\
				PL model for NLOS communications,  $\mathrm{PL}_{\mathtt{NLOS}}$& 145.4 + 37.5$\log_{10}(d)$ dB\\
				Power budget at the FD-BS, $P_{bs}^{\max}$     & 26 dBm\\
				Power budget at  UL users, $P_{i,\ell}^{\max} \equiv P_{\mathtt{U}}^{\max}$     & 23 dBm \\
			  Degree of residual SI, $\sigma_{\mathtt{SI}}$ & -75 dB\\
				Number of antennas at  the FD-BS, $N_{\mathtt{t}} = N_{\mathtt{r}}$ & 5\\
				Threshold of all UL users, $\bar{\mathtt{R}}_{i,\ell}^{\mathtt{U}}\equiv \bar{\mathtt{R}}^{\mathtt{U}}$ & 2 bps/Hz\\
		\hline		   				
		\end{tabular}}
\end{table}

\begin{figure}[t]
\centering
        \includegraphics[width=0.40\textwidth]{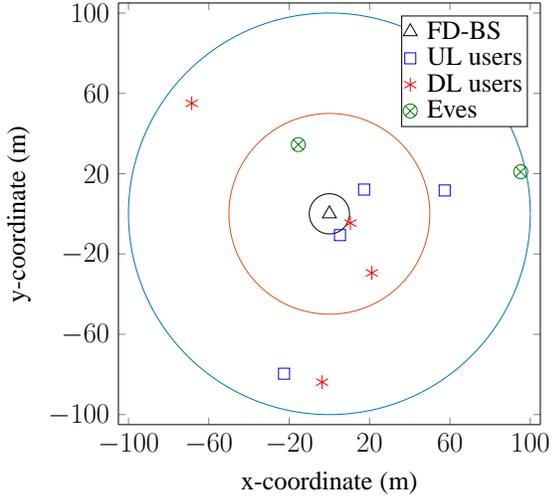}
	  \caption{A small cell topology with four DL users $(K=2)$, four UL users $(L=2)$ and $M=2$ Eves is used in the numerical examples. The radius of the small cell is set to  100 m with an inner circle radius of 50 m. Two DL users and two UL users are randomly located in zone-1 (inner zone) and the remaining two DL users and two UL users are randomly located in  zone-2 (outer zone). There is one $N_{e,m} = 2$-antenna Eve that is randomly placed in each zone. }\label{fig:celllayout}
\end{figure}

We now  evaluate the performance of the proposed FD model using realistic parameters. We consider the system topology illustrated in Fig. \ref{fig:celllayout}. There are two DL users, two UL users, and one $N_{e,m} = 2$-antenna Eve placed in each zone for a small cell topology. Unless stated otherwise,  most important parameters following FD radio evaluation methodology \cite{Bharadia13,Duarte:TWC:12} agreed in the 3GPP \cite{3GPP} are specified in Table~\ref{parameter}  for ease of cross-referencing. Here, the power budgets of all UL users are set to be equal, i.e. $P_{i,\ell}^{\max} \equiv P_{\mathtt{U}}^{\max}$. The entries of the fading loop channel $\mathbf{G}_{\mathtt{SI}}$ are  generated as independent and identically distributed Rician random variables with Rician factor $K_{{\mathtt{SI}}}=5$ dB. The channel  from  UL user ($i,\ell$) to DL user ($i,k$) (CCI) at a distance $d$  in km is assumed to undergo  path loss (PL) model for non-line-of-sight (NLOS) communications as $f_{i,k,\ell} = \sqrt{10^{-\mathrm{PL}_{\mathtt{NLOS}}/10}}\tilde{f}_{i,k,\ell}$, where $\mathrm{PL}_{\mathtt{NLOS}}$ is the PL in dB and $\tilde{f}_{i,k,\ell}$ follows $\mathcal{CN}(0,1)$ representing small-scale effects \cite{Dan:TWC:14,3GPP}. All other channels follow the PL model for line-of-sight (LOS) communications as $\mathbf{L} = \sqrt{10^{-\mathrm{PL}_{\mathtt{LOS}}/10}}\tilde{\mathbf{L}}$, where $\mathbf{L}\in\{\mathbf{h}_{i,k}, \mathbf{g}_{i,\ell},\mathbf{H}_m,\mathbf{g}_{m,i,\ell}\}$ and the entries of $\tilde{\mathbf{L}}$ follow $\mathcal{CN}(0,1)$. Herein, we have favored the channel quality of Eves due to their capability to select a good location to avoid a high possibility of obstruction.
The error tolerance between two consecutive iterations  in the proposed Algorithms is set to $\epsilon  = 10^{-3}$. The achieved SR results in nats/sec/Hz are divided by $\ln(2)$ to have at  units of bps/channel-use. We use MOSEK as the convex solver with  the toolbox YALMIP \cite{Lofberg} in the MATLAB environment.  Results  are obtained by averaging over 1,000 runs.

For comparison, the following three existing schemes are considered:
\begin{itemize}
   \item  We consider  a conventional FD system, under which all DL and UL users are simultaneously served in the whole communication time block (i.e., without considering fractional times and user grouping \cite{SunTWC16,Dan:TWC:14}). We call this scheme ``Conventional FD.''
	\item Under the same system model with ``Conventional FD,'' the DL transmission can adopt NOMA \cite{DSP16,Choi15,Nguyen:JSAC:17} to further improve its  performance. Here, each DL user in zone-1 is paired with a DL user in zone-2 to create a virtual cluster by using the clustering algorithm in  \cite{Choi15}. In each virtual cluster, the message intended for the DL user in zone-2 is decoded by both users while  the message intended for the DL user in zone-1 is  successively decoded by itself after processing SIC to cancel the interference of the former \cite{Nguyen:JSAC:17}. Since the  FD-BS already adopted MMSE-SIC decoder for the UL reception of UL users, we call this scheme ``FD-NOMA.'' Note that the achievable rate of DL users can be improved by considering a larger cluster size  \cite[Sec. VI-D]{Nguyen:JSAC:17}, but their information privacy is more exposed \cite{NguyenCLFT17}.
	\item Additionally, an HD system is considered. Here, HD-BS uses all available antennas $N = N_{\mathtt{t}}+N_{\mathtt{r}}$ and the same power budget with FD-BS to serve all DL users in the DL and all UL users in the UL, albeit in two separate  communication time blocks. In such a case,  there are no SI and CCI, but the effective SR suffers from a reduction by  half.
\end{itemize}
To conduct a fair comparison, the BS in all those schemes also injects an AN in  DL transmission and adopts an MMSE-SIC decoder for the reception of the UL signals with the same decoding order as previously presented. In parallel to the max-min SR among all DL  and UL users (e.g., \eqref{eq:OP1}), we also plot the max-min SR for DL users only but subject to UL users' SR requirements (i.e., $\bar{\mathtt{R}}_{i,\ell}^{\mathtt{U}}\equiv \bar{\mathtt{R}}^{\mathtt{U}}, \forall (i,\ell)\in\mathcal{U}$ listed in Table \ref{parameter}) as mentioned in Remark \ref{Remark1}. For convenience, the SR obtained by the former is referred to as ``max-min SR'' while that of the latter (i.e., by \eqref{eq:OP1a}) is referred to as  ``max-min SR of DL users.'' It is obvious that the optimization problems  corresponding to the above schemes can be solved using our proposed methods after some modifications.

\vspace{-0.2cm}
\subsection{Numerical Results for  Known CSI \eqref{eq:OP1}}

\begin{figure}
    \begin{center}
				\begin{subfigure}[Average max-min SR versus $P_{bs}^{\max}$.]{
        \includegraphics[width=0.42\textwidth]{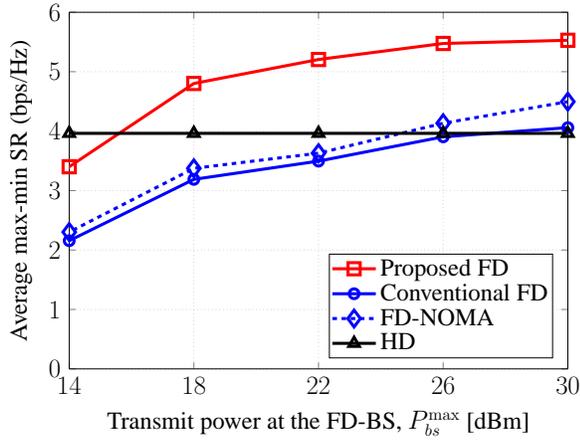}}
    		\label{fig:SRPCSI:PvsandSI:a}
				\end{subfigure}
				\begin{subfigure}[Average max-min SR versus $\sigma_{\mathtt{SI}}$.]{
        \includegraphics[width=0.42\textwidth]{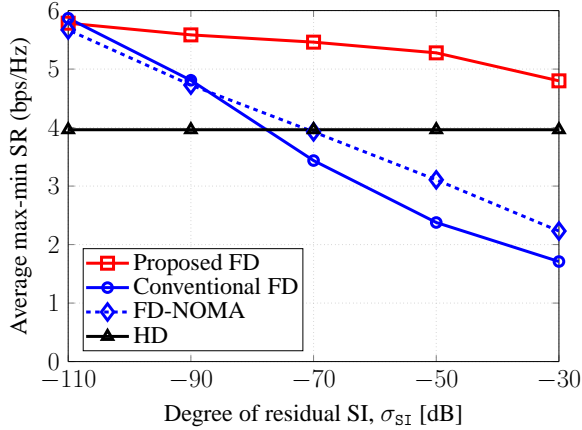}}
    		\label{fig:SRPCSI:PvsandSI:b}
				\end{subfigure}
	  \caption{Average max-min SR with known CSI (a) versus the transmit power at the FD-BS and (b) versus the degree of residual SI.}\label{fig:SRPCSI:PvsandSI}
\end{center}
\end{figure}

Fig.~\ref{fig:SRPCSI:PvsandSI}(a) depicts the average max-min SR versus the FD-BS transmit power  with known CSI for different resource allocation schemes.
The observations from the figure are as follows. First, one can see that the SR for the FD systems is better than that of HD system at a high transmit power $P_{bs}^{\max}$, and the SR curve of HD is nearly unchanged. The reasons for these results are three-fold: 1) The effective SR of the HD per resource block is divided by two; 2) The DL transmission in HD dominates the UL one, as the UL transmission is free of AN and due to the effectiveness of DL beamforming; 3) In the FD systems,  FD-BS can better protect both  DL and UL transmissions by exploiting  DL interference. Second, the SR of FD-NOMA outperforms conventional FD, which is a result
of canceling out intra-cluster interference. Third, the SR of the proposed FD is fully superior  to  the other schemes and an improvement of almost 1.51 bps/Hz ($\approx$ 38.2$\%$) over HD  is achieved at the practical value of $P_{bs}^{\max} = 26$ dBm, which is defined in 3GPP TS 36.814. We recall that the proposed FD is advantageous over other schemes in terms of handling interference.

In Fig.~\ref{fig:SRPCSI:PvsandSI}(b), we plot the average max-min SR as a function of $\sigma_{\mathtt{SI}}$. We can see that the proposed FD scheme offers a significant gain over  traditional FD schemes in terms of achievable SR. When $\sigma_{\mathtt{SI}}$ becomes stringent, it is even more essential. Although a major part of the power budget must be allocated to serve far DL users in improving their
SR leading to a significant effect of the SI, the near UL users' SR requirement is easily met. In other words, the effect of SI in the proposed FD becomes less due to the effectiveness of user grouping method. As $\sigma_{\mathtt{SI}}$ increases, the SRs of conventional FD and FD-NOMA  drop quickly and tend to be worse than those of HD when $\sigma_{\mathtt{SI}}$ is higher than a certain level. These results are probably attributed to the fact that
the FD-BS in those schemes needs to scale down its transmit power to maintain the UL users' QoS, which leads to a significant loss in the system performance. In addition, Fig.~\ref{fig:SRPCSI:PvsandSI}(b) further shows that the degree of residual SI needs to be suppressed by at least 72 dB for the FD-NOMA and at least 78 dB for the conventional FD to guarantee a better SR per-user compared to  HD. Interestingly, the SR of the proposed FD outperforms the HD for a broad range of $\sigma_{\mathtt{SI}}$, which confirms its robustness against the significant effect of SI.

\begin{figure}
    \begin{center}
				\begin{subfigure}[Average max-min SR of DL users versus $\bar{\mathtt{R}}^{\mathtt{U}}$ for $N_{\mathtt{t}}=N_{\mathtt{r}}=5$.]{
        \includegraphics[width=0.42\textwidth]{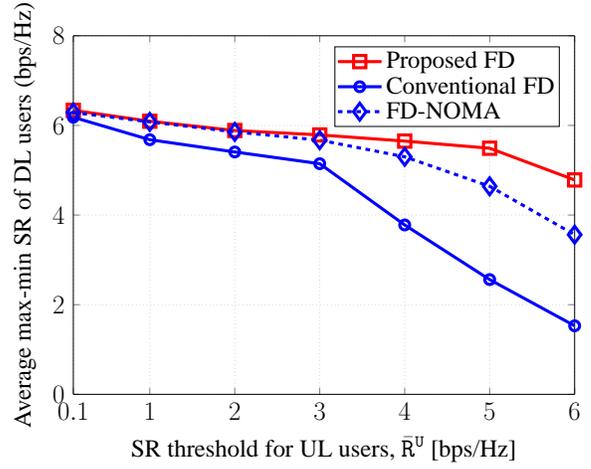}}
    		\label{fig:SRPCSI:Ru5Nt}
				\end{subfigure}
				\begin{subfigure}[Average max-min SR of DL users versus $\bar{\mathtt{R}}^{\mathtt{U}}$ for $N_{\mathtt{t}}=N_{\mathtt{r}}=3$.]{
        \includegraphics[width=0.42\textwidth]{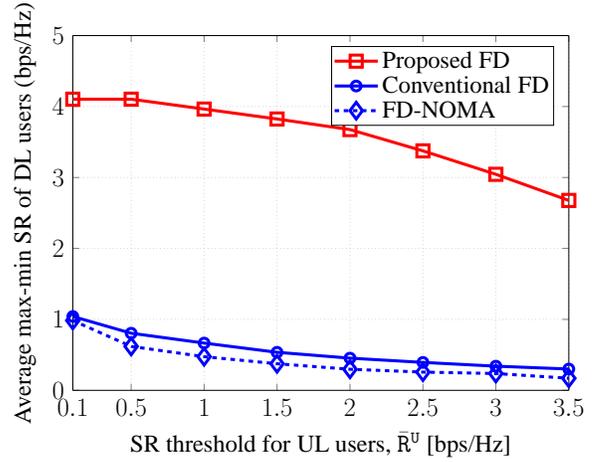}}
    		\label{fig:SRPCSI:Ru3Nt}
				\end{subfigure}
	  \caption{Average max-min SR of DL users versus the SR threshold for UL users with known CSI (a) for  $N_{\mathtt{t}}=N_{\mathtt{r}}=5$ and (b) for $N_{\mathtt{t}}=N_{\mathtt{r}}=3$.}\label{fig:SRPCSI:Ru}
\end{center}
\end{figure}

A trade-off of the average max-min SR between DL  and UL users is illustrated in Fig. \ref{fig:SRPCSI:Ru} for different numbers of antennas at the FD-BS by solving \eqref{eq:OP1a}. The results of the performance  for HD are not shown here due to the independence of two resource blocks. In Fig. \ref{fig:SRPCSI:Ru}(a) for $N_{\mathtt{t}}=N_{\mathtt{r}}=5$, we observe that the SRs of all schemes  constantly decrease with an increase in $\bar{\mathtt{R}}^{\mathtt{U}}$ since the feasible set gets more restricted. For a high demand on the UL transmission,  FD-BS must  scale down its transmit power to avoid severe  interference to the UL reception, leading to a drastic reduction in the SR for the DL transmission.  As expected, the proposed FD scheme achieves better performance in terms of the SR compared to the others, especially in the  range of $\bar{\mathtt{R}}^{\mathtt{U}}\geq 4$ bps/Hz.  For $N_{\mathtt{t}}=N_{\mathtt{r}}=3$, Fig. \ref{fig:SRPCSI:Ru}(b)  demonstrates the advantage of the proposed FD scheme when the number of users is larger than the number of transmit/receive antennas. As can be seen,  traditional FD schemes cannot deliver a good SR for the given setup, owing to  a lack of the DoF to leverage multiuser diversity (recall  $N_{\mathtt{t}} = 3 < 4$ DL users and $N_{\mathtt{r}} = 3 < 4$ UL users).  The UL users' QoS ability  of the FD-NOMA is inferior to the conventional FD due to inefficiency of using SIC in this setting. In contrast, the proposed FD scheme still has sufficient DoF to transmit and receive the information signals in both directions (at a specific time,  FD-BS concurrently serves only 2 DL  and 2 UL users). Consequently, this results in a substantial improvement, of about 3.37 bps/Hz and 3.22 bps/Hz at $\bar{\mathtt{R}}^{\mathtt{U}}$ = 2 bps/Hz, in the SR of the proposed FD scheme when compared with FD-NOMA and  conventional FD, respectively.

\subsection{Numerical Results for SCSI of Eves \eqref{eq:OP7}}

\begin{figure}
    \begin{center}
				\begin{subfigure}[Average max-min SR versus $P_{bs}^{\max}$.]{
        \includegraphics[width=0.42\textwidth]{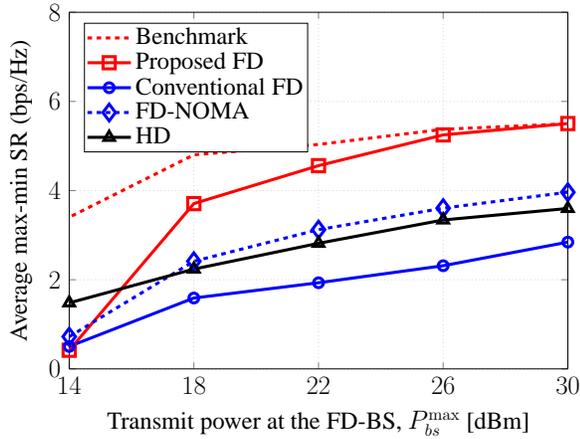}}
    		\label{fig:SCSIvsPbs:a}
				\end{subfigure}
				\begin{subfigure}[Average max-min SR of DL users versus $P_{bs}^{\max}$ for $\bar{\mathtt{R}}^{\mathtt{U}} = 2$ bps/Hz.]{
        \includegraphics[width=0.42\textwidth]{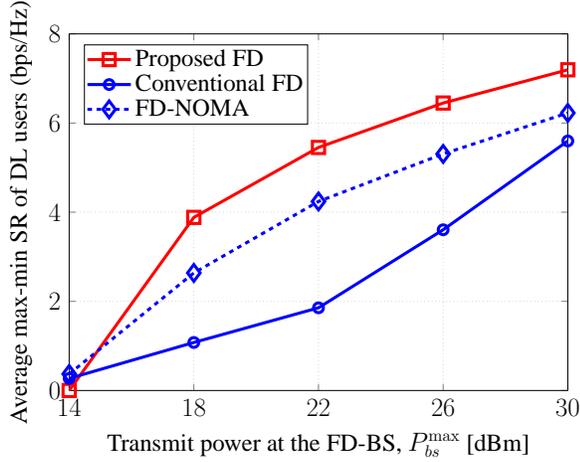}}
    		\label{fig:SCSIvsPbs:b}
				\end{subfigure}
	  \caption{(a) Average max-min SR per-user and (b) average max-min SR of DL users for $\bar{\mathtt{R}}^{\mathtt{U}} = 2$ bps/Hz, versus the transmit power at the FD-BS with SCSI for Eves.}\label{fig:SCSIvsPbs}
\end{center}
\end{figure}

In this subsection, we show numerical results for the SRM-SCSI problem \eqref{eq:OP7}. Under the same simulation setup as in the previous subsection, we set $\epsilon_{i,k}  = 0.99, \forall (i,k)\in\mathcal{D}$ and $\tilde{\epsilon}_{i,\ell} = 0.99, \forall (i,\ell)\in\mathcal{U}$ to guarantee secure communications in both directions. In Fig.~\ref{fig:SCSIvsPbs}(a), we also plot a benchmark  of the proposed FD scheme, assuming  perfect CSI for the Eves. The results in Fig.~\ref{fig:SCSIvsPbs}(a) reveal that the SRs of all resource allocation schemes degrade compared to Fig.~\ref{fig:SRPCSI:PvsandSI}(a), which is due to a lack of CSI on Eves. In other words, the perfectly known CSI at the transmitters plays a vital role in designing  effective beamforming. Otherwise, this will result in the cost of a reduced system performance. Notably, the SRs of the proposed FD and  FD-NOMA schemes have fewer loss than the others due to  efficient proposed design and SIC, respectively.  The SR of the proposed FD scheme also approaches  that of the benchmark  when $P_{bs}^{\max}$ increases. This is because the proposed FD scheme aims to manage the network interference to improve the SR rather than concentrating the interference at Eves. At $P_{bs}^{\max} = 26$ dBm, significant gains of up to 126.8$\%$, 57.1$\%$ and 45.5$\%$ are offered by the proposed FD scheme compared to  conventional FD,  HD and  FD-NOMA, respectively. These results confirm that the proposed FD scheme is more robust and reliable in the presence of imperfect CSI of Eves compared to the others.

The average max-min SR of the DL users versus the FD-BS transmit power  is given in Fig.~\ref{fig:SCSIvsPbs}(b) for $\bar{\mathtt{R}}^{\mathtt{U}} = 2$ bps/Hz. As can be observed, the SRs of all schemes grow very rapidly when $P_{bs}^{\max}$  increases. The reasons behind this behavior are as follows. 1) The UL users can easily tune the power  in meeting their
 QoS requirements (a loose QoS requirement) to avoid strong CCI to the DL users; 2) The FD-BS now pays more attention to serve the DL users by transferring more power to them with  less attention to the SI. Again, the proposed FD scheme achieves much better SR compared to the traditional FD schemes.

\subsection{Numerical Results for Worst-Case Scenario \eqref{eq:OP9}}

\begin{figure}
    \begin{center}
				\begin{subfigure}[Average max-min SR versus the number of antennas at the FD-BS, $N_{\mathtt{t}}=N_{\mathtt{r}}$.]{
        \includegraphics[width=0.41\textwidth]{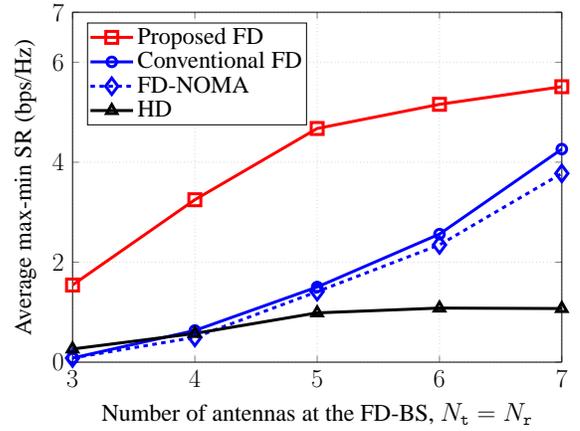}}
    		\label{fig:WCSvsN:a}
				\end{subfigure}
				\begin{subfigure}[Average max-min SR versus the number of users per zone, $K = L$.]{
        \includegraphics[width=0.41\textwidth]{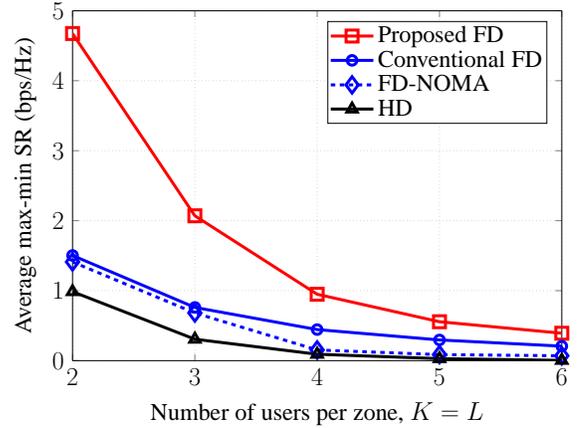}}
    		\label{fig:WCSvsKL:b}
				\end{subfigure}
	  \caption{Average max-min SR with worst-case scenario (a) versus the number of antennas at the FD-BS and (b) versus the number of users per zone.}\label{fig:WCSvsNKL}
\end{center}
\end{figure}

\begin{figure}[t]
\centering
        \includegraphics[width=0.41\textwidth]{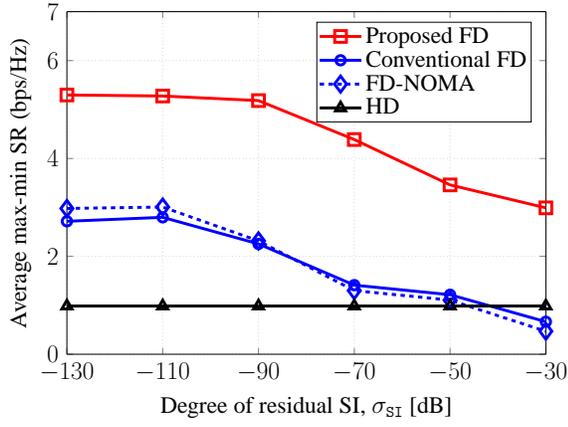}
	  \caption{Average max-min SR with worst-case scenario versus the degree of residual SI.}\label{fig:WCSvsSI}
\end{figure}

For the SRM-WCS \eqref{eq:OP9}, we plot the  average max-min SR  versus the number of antennas at the FD-BS in Fig.~\ref{fig:WCSvsNKL}(a) and versus the number of users per zone in Fig.~\ref{fig:WCSvsNKL}(b). As expected,  the SRs of all resource allocation
schemes shown in Fig.~\ref{fig:WCSvsNKL}(a)  degrade (i.e., compared to Fig.~\ref{fig:SRPCSI:PvsandSI}(a) at $N_{\mathtt{t}}=N_{\mathtt{r}}=5$), and this is even more  drastic for  HD. It is easy to see that the UL-user-to-Eve links make the HD scheme more vulnerable to interception than the FD ones, since  the UL transmission is free of both MUI and AN. In contrast,  FD-BS with AN can better protect the information signals in both directions, which further confirms the importance of using AN. As seen, for $N_{\mathtt{t}}=N_{\mathtt{r}} < 4$, the SRs of  FD-NOMA and conventional FD are inferior to  HD; however, as $N_{\mathtt{t}}$ and $N_{\mathtt{r}}$ increase, the benefit of exploiting FD radio outweighs the lack of the DoF to leverage multiuser diversity. Another interesting observation is that the SR of FD-NOMA is slightly worse  than that of the conventional FD. The reason for this is because the FD-BS must allocate  a major part of the power budget to produce AN leading to  less power to convey the desired signals, which may cause  error propagation in SIC.  In Fig.~\ref{fig:WCSvsNKL}(b),  we see that increasing the number of users
severely deteriorates the SR of all schemes. For fixed dimensionality of the beamforming vectors and large number of users, the BSs are not able to  suppress the MUI effectively and the power allocated to each DL user is significantly reduced, compromising the SR per user. On the other hand, the SRs of FD-NOMA and HD are close to zero for $K = L \geq 4$, which can be  intuitively explained as follows. In order to perform SIC effectively in FD-NOMA, FD-BS  must align the transmit signal of far DL users around near DL users, which in turn may result in strong interference at the unintended DL users and Eves in zone-1. In HD, increasing the number of UL users also makes a higher probability of the signal leakage into Eves. Nevertheless, the proposed FD scheme still achieves the best SR by exploiting the spatial  DoF more efficiently.

Fig.~\ref{fig:WCSvsSI} shows the impact of the degree of residual SI on the performance of the system. We observe that the  residual SI needs to be canceled no less than 45 dB (i.e., $\sigma_{\mathtt{SI}} < -45$ dB) for traditional FD schemes to achieve a better SR compared to  HD. From Fig.~\ref{fig:WCSvsSI} and recalling the
discussions presented for Fig.~\ref{fig:SRPCSI:PvsandSI}(b), it can be seen that the proposed FD scheme can  utilize the transmission power more efficiently, and thus, it achieves remarkable gains compared to the others.

\subsection{The Importance of Using AN}
To examine when AN is important for the proposed designs, we measure the percentage of AN's power to the total transmit power at the FD-BS, $P_{bs}^{\max}$, which is defined as $\frac{\sum_{i=1}^2\tau_i\|\mathbf{V}_i\|_\F^2}{P_{bs}^{\max}}\times 100\%$.

\begin{figure}[t]
\centering
        \includegraphics[width=0.465\textwidth]{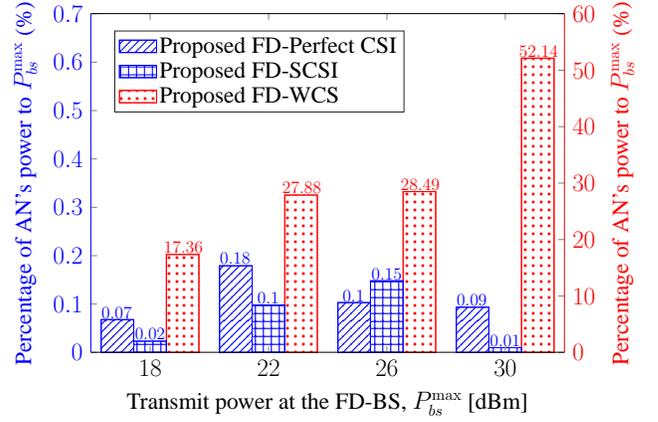}
	  \caption{The percentage of AN's power to $P_{bs}^{\max}$  versus $P_{bs}^{\max}$  for three proposed FD designs: perfect CSI and Eves' SCSI (left hand y-axis), and WCS (right hand y-axis).}\label{fig:PerANtoPbsvsPbs}
\end{figure}

In Fig.~\ref{fig:PerANtoPbsvsPbs}, we show the percentage of AN's power versus  the FD-BS transmit power for three proposed FD designs with the problem of max-min SR. As seen, for the proposed FD designs with perfect CSI and Eves' SCSI, the use of AN is not crucial. A very small portion (i.e., less than 0.18$\%$) of the total transmit power is allocated to AN. Since the SR of DL users is virtually higher than that of UL users, the max-min SR will be determined by the UL users. For these designs, the DL multiuser interference can effectively debilitate the Eves' reception, and there is no need to allocate a high power on AN. However, AN becomes important for   the proposed FD design with WCS. When the FD-BS transmit power increases,  the percentage of AN's power  increases significantly, i.e.,  28.49$\%$ at $P_{bs}^{\max} = 26$ dBm and  reaching 52.14$\%$ at $P_{bs}^{\max} = 30$ dBm.  In this case, Eves are capable of suppressing the
AN by the MMSE decoder and canceling all the MUI, and thus, more power should be allocated to AN to achieve higher SR. On the other hand, the optimal value for the FT $\tau_i, i=1,2$ is about 0.5 on the average for this symmetric setting, but can be changed depending  on the demand of each user group.

\begin{figure}[t]
\centering
        \includegraphics[width=0.40\textwidth]{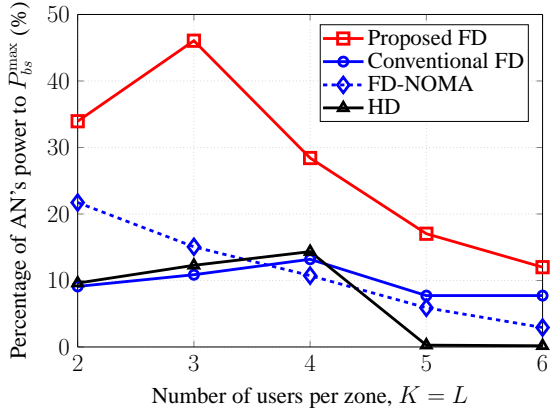}
	  \caption{The percentage of AN's power to $P_{bs}^{\max}$  with worst-case scenario versus the number of users per zone.}\label{fig:PerANtoPbsvsKL}
\end{figure}

Under the same setting with worst-case scenario as Fig.~\ref{fig:WCSvsNKL}(b), we show the percentage of AN's power versus the number of users per zone in Fig.~\ref{fig:PerANtoPbsvsKL}. The percentage of AN's power of all schemes decreases from a certain value of $K$ and $L$. For large $K$ and $L$, more power needs to be allocated to information symbols to guarantee a high power received at DL users, resulting in less power on  AN. Interestingly,  the proposed FD scheme with more DoF can still allocate  a suitable portion of power budget to produce AN to achieve the best SR, as shown in Fig.~\ref{fig:WCSvsNKL}(b).

\subsection{Algorithm Convergence}
\begin{figure}[t]
\centering
        \includegraphics[width=0.40\textwidth]{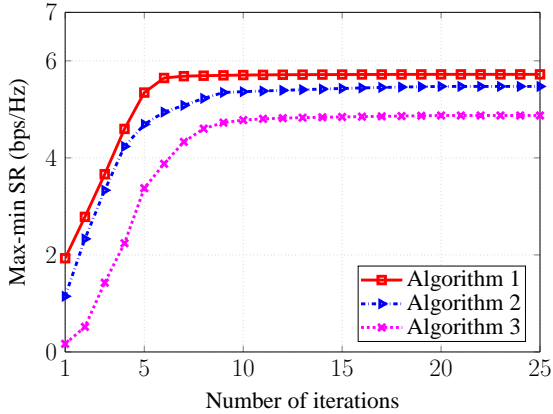}
	  \caption{Convergence of the algorithms.  }\label{fig:Convergence}
\end{figure}

Fig. \ref{fig:Convergence} plots the typical convergence results of Algorithms \ref{algo:knownCSI} to \ref{algo:worstcase}  for  a randomly generated channel realization. We can see that all the proposed algorithms converge very fast to the optimal solution within  about 10 iterations. The optimization 
variables are accounted for to find a better solution for the next iteration.   Intuitively, the SRs of all algorithms
increase quickly within the first iterations and stabilize after a few more iterations. This is because in the first iterations, the approximation errors are large. However, when the algorithms reach  the optimal solution, the approximation errors  become small due to updating  the involved optimization variables after each iteration.

\section{Conclusion}\label{Conclusion}

We have addressed the problem of secure FD multiuser wireless communication. To manage the network interference more  effectively, a simple and very efficient user grouping-based fractional time model has been proposed. Depending on how much the transmitters know about the Eves'
CSI,  three difficult nonconvex optimization problems have been considered: $(i)$ SRM with known CSI, $(ii)$ SRM with only  Eves' SCSI  and $(iii)$ SRM with WCS.  We have developed new path-following optimization algorithms to jointly design the fractional times and power resource allocation to maximize the SR per-user in both DL and UL directions. Specifically, we have first transformed the nonconvex optimization problem to a tractable form and then solved a  sequence of convex programs with polynomial computational complexity in each iteration. Numerical results with realistic parameters have confirmed  fast convergence to at least local optima of the original nonconvex design problems. They have revealed that the proposed FD scheme not only provides substantial  improvement in terms of SR
when compared to the known solutions  (i.e.,  HD, conventional FD and FD-NOMA), but is also robust  against the significant effects of SI, imperfect CSI on Eves  and  DoF bottleneck.

\section*{Appendix~A\\ \quad\ Proof of Inequality \eqref{inq1}} \label{Appendix:A}
\renewcommand{\theequation}{A.\arabic{equation}}
\setcounter{equation}{0}
Considering the function $\zeta(z,t) \triangleq \ln(1+1/z)/t$, we have $\zeta(z,t) = -\ln\bigl(1-1/(1+z)\bigr)/t$. It is clear that  $-\ln\bigl(1-1/(1+z)\bigr)$ is convex function and decreasing on the domain $z > 0$, while the function $1/t$ is convex on the domain $t > 0$. Therefore, the composite function $\zeta(z,t)$ is convex on the domain $z>0,t>0$ \cite{Nguyen:TCOM:17,Tuybook}. Thus, it is true that \cite{Tuybook}: 
\begin{eqnarray}\label{eq:A1}
\zeta(z,t)&\geq&\ds \zeta(z^{(\kappa)},t^{(\kappa)}) + \Bigl\langle\bigl[\nabla_{z}\zeta(z,t)|_{(z^{(\kappa)},t^{(\kappa)})},\ \nonumber\\
&&\nabla_{t}\zeta(z,t)|_{(z^{(\kappa)},t^{(\kappa)})}\bigr]\bigl[(z-z^{(\kappa)}),\ (t-t^{(\kappa)})\bigr]^T \Bigr\rangle\nonumber\\
&=& 2\zeta(z^{(\kappa)},t^{(\kappa)}) + \frac{1}{t^{(\kappa)}(z^{(\kappa)}+1)}  \nonumber\\
&& \quad-\; \frac{1}{t^{(\kappa)}z^{(\kappa)}(z^{(\kappa)}+1)}z - \frac{\zeta(z^{(\kappa)},t^{(\kappa)})}{t^{(\kappa)}}t
\end{eqnarray}
where the notation $(\cdot)|_{(z^{(\kappa)},t^{(\kappa)})}$ is used to represent the value of the function at ${(z^{(\kappa)},t^{(\kappa)})}$.  The inequality \eqref{inq1} is then obtained by substituting $\gamma = z^{-1}$ and $\gamma^{(\kappa)} = (z^{(\kappa)})^{-1}$ into \eqref{eq:A1}.

\section*{Appendix~B\\ \quad  Derivations of $\mathtt{F}_{m,i,k}^{(\kappa)}(\mathbf{X}_i,\alpha_i,\mu_{m,i,k})$ and $\mathtt{P}_{m,i,\ell}^{(\kappa)}(\mathbf{X}_i,\alpha_i,\tilde{\mu}_{m,i,\ell})$ } \label{Appendix:B}
\renewcommand{\theequation}{B.\arabic{equation}}
\setcounter{equation}{0}
For the concave function $\sqrt{yz}$, its convex upper bound can be found as \cite{Beck:JGO:10}
\begin{eqnarray}
\sqrt{yz} &\leq& \frac{\sqrt{y^{(\kappa)}}}{2\sqrt{z^{(\kappa)}}}z + \frac{\sqrt{z^{(\kappa)}}}{2\sqrt{y^{(\kappa)}}}y \label{B3}
\end{eqnarray}
with $\forall y > 0, y^{(\kappa)} > 0, z > 0, z^{(\kappa)} > 0$. The convex approximation  of $\mathtt{F}_{m,i,k}(\mathbf{X}_i,\alpha_i,\mu_{m,i,k})$
can be found as follows. The  first term $\mu_{m,i,k}/\alpha_i$ is convexified by using \eqref{B3} while the second term  is a quadratic function, which can be linearized for inner approximation by using \eqref{eq:inequad}.  As a result, we have
\begin{IEEEeqnarray}{rCl}
&& \mathtt{F}_{m,i,k}^{(\kappa)}(\mathbf{X}_i,\alpha_i,\mu_{m,i,k})\triangleq \frac{1}{2}\Bigl(\frac{\mu_{m,i,k}^2}{\mu_{m,i,k}^{(\kappa)}\alpha_i^{(\kappa)}} + \frac{\mu_{m,i,k}^{(\kappa)}}{2\alpha_i-\alpha_i^{(\kappa)}}\Bigr)  -\nonumber\\ 
&&\ \left[\mathtt{Q}_{m,i}^{(\kappa)}(\mathbf{X}_i) - 2\Re\bigl\{\bigl(\mathbf{w}_{i,k}^{(\kappa)}\bigr)^H\mathbf{H}_{m}\mathbf{H}_{m}^H\mathbf{w}_{i,k}\bigl\} + \|\mathbf{H}_{m}^H\mathbf{w}_{i,k}^{(\kappa)}\|^2\right]\label{B4}\qquad
\end{IEEEeqnarray}
where 
\begin{IEEEeqnarray}{rCl}
&& \mathtt{Q}_{m,i}^{(\kappa)}(\mathbf{X}_i)\triangleq 2\Biggl(\sum_{k'=1}^K \Re\Bigl\{\bigl(\mathbf{w}_{i,k'}^{(\kappa)}\bigr)^H\mathbf{H}_{m}\mathbf{H}_{m}^H\mathbf{w}_{i,k'}\Bigl\}\; +\nonumber\\
&& \Re\Bigl\{\tr\Bigl(\bigl(\mathbf{V}_{i}^{(\kappa)}\bigr)^H\mathbf{H}_{m}\mathbf{H}_{m}^H\mathbf{V}_{i}\Bigr)\Bigl\} + \sum_{\ell^{'}=1}^L\rho_{i,\ell^{'}}^{(\kappa)}\rho_{i,\ell^{'}}\|\mathbf{g}_{m,i,\ell^{'}}^H\|^2\Biggr)\; - \nonumber\\
&&\Bigl(\sum_{k'=1}^K\|\mathbf{H}_{m}^H\mathbf{w}_{i,k'}^{(\kappa)}\|^2+\|\mathbf{H}_{m}^H\mathbf{V}_{i}^{(\kappa)}\|_\F^2
 +\sum_{\ell'=1}^L\bigl(\rho_{i,\ell^{'}}^{(\kappa)}\bigr)^2\|\mathbf{g}_{m,i,\ell^{'}}^H\|^2\Bigr),\quad  \nonumber
\end{IEEEeqnarray}
and $\alpha_i^2$ is  linearized as $\alpha_i^{(\kappa)}(2\alpha_i-\alpha_i^{(\kappa)})$.
Analogously, the function  $\mathtt{P}_{m,i,\ell}(\mathbf{X}_i,\alpha_i,\tilde{\mu}_{m,i,\ell})$ is approximated by
\begin{IEEEeqnarray}{rCl}
 &&\mathtt{P}_{m,i,\ell}^{(\kappa)}(\mathbf{X}_i,\alpha_i,\tilde{\mu}_{m,i,\ell})\triangleq \frac{1}{2}\Bigl(\frac{\tilde{\mu}_{m,i,\ell}^2}{\tilde{\mu}_{m,i,\ell}^{(\kappa)}\alpha_i^{(\kappa)}} + \frac{\tilde{\mu}_{m,i,\ell}^{(\kappa)}}{2\alpha_i-\alpha_i^{(\kappa)}}\Bigr)  -\nonumber\\ 
&&\ \left[\mathtt{Q}_{m,i}^{(\kappa)}(\mathbf{X}_i) - 2\rho_{i,\ell}^{(\kappa)}\rho_{i,\ell}\|\mathbf{g}_{m,i,\ell}^H\|^2 + \bigl(\rho_{i,\ell}^{(\kappa)}\bigr)^2\|\mathbf{g}_{m,i,\ell}^H\|^2\right].    \label{B5}\qquad
\end{IEEEeqnarray}
 
\section*{Appendix~C\\ \quad Proof of Proposition \ref{pro:1} } \label{Appendix:C}
\renewcommand{\theequation}{C.\arabic{equation}}
\setcounter{equation}{0}
Let us prove  constraint \eqref{eq:changeOP3:a} corresponding to its approximation \eqref{eq:R1kConvex} first.
It holds that
\begin{IEEEeqnarray}{rCl}\label{eq:C1}
C_{i,k}^{\mathtt{D},(\kappa)}(\mathbf{X}_i,\alpha_i) \geq \eta + \Gamma_{i,k}^{\mathtt{D}}.
\end{IEEEeqnarray}
For any feasible point $\bigl(\mathbf{X}_i^{(\kappa)},\alpha_i^{(\kappa)}\bigr)\in\mathcal{V}^{(\kappa)}$ of \eqref{eq:changeOP3:a}, i.e.,
$C_{i,k}^{\mathtt{D}}\bigl(\mathbf{X}_i^{(\kappa)},\alpha_i^{(\kappa)}\bigr) \geq \eta + \Gamma_{i,k}^{\mathtt{D}}$,
it follows that 
\begin{IEEEeqnarray}{rCl}\label{eq:C3}
C_{i,k}^{\mathtt{D},(\kappa)}(\mathbf{X}_i^{(\kappa)},\alpha_i^{(\kappa)}) \geq \eta + \Gamma_{i,k}^{\mathtt{D}}\ (\text{due to \eqref{eq:R1kAgree}}).
\end{IEEEeqnarray}
This implies that $\bigl(\mathbf{X}_i^{(\kappa)},\alpha_i^{(\kappa)}\bigr)\in\mathcal{V}^{(\kappa)}$ is also feasible for \eqref{eq:C1} (and hence \eqref{eq:R1kConvex}). Therefore, the optimal solution $\bigl(\mathbf{X}_i^{(\kappa+1)},\alpha_i^{(\kappa+1)}\bigr)$ of \eqref{eq:C1} should satisfy \eqref{eq:changeOP3:a} because
\begin{IEEEeqnarray}{rCl}\label{eq:C4}
&&C_{i,k}^{\mathtt{D},(\kappa)}\bigl(\mathbf{X}_i^{(\kappa+1)},\alpha_i^{(\kappa+1)}\bigr) \geq \eta + \Gamma_{i,k}^{\mathtt{D}}\nonumber\\
\Leftrightarrow\; &&C_{i,k}^{\mathtt{D}}\bigr(\mathbf{X}_i^{(\kappa+1)},\alpha_i^{(\kappa+1)}\bigl) \geq \eta + \Gamma_{i,k}^{\mathtt{D}}.
\end{IEEEeqnarray}
The above result holds true for the remaining nonconvex constraints and their convex approximations whenever Algorithm \ref{algo:knownCSI} is initialized with $\bigl(\mathbf{X}^{(0)},\boldsymbol{\alpha}^{(0)},\boldsymbol{\mu}^{(0)}\bigr)$ satisfying the feasibility conditions of \eqref{eq:OP4}.

\section*{Appendix~D\\ \quad\  Proof of Proposition \ref{pro:2} } \label{Appendix:D}
\renewcommand{\theequation}{D.\arabic{equation}}
\setcounter{equation}{0}
For the first point of Proposition \ref{pro:2}, we focus on \eqref{eq:changeOP3:a} and the same arguments will be applicable to all remaining constraints. From \eqref{eq:Rate1ConvexAppro} and \eqref{eq:R1kAgree}, we recall that
\begin{IEEEeqnarray}{rCl}\label{eq:D1}
C_{i,k}^{\mathtt{D}}(\mathbf{X}_i,\alpha_i) &\geq& C_{i,k}^{\mathtt{D},(\kappa)}(\mathbf{X}_i,\alpha_i)\ \text{and}\ \nonumber\\
C_{i,k}^{\mathtt{D}}\bigl(\mathbf{X}_i^{(\kappa)},\alpha_i^{(\kappa)}\bigl) &=& C_{i,k}^{\mathtt{D},(\kappa)}\bigl(\mathbf{X}_i^{(\kappa)},\alpha_i^{(\kappa)}\bigr) 
\end{IEEEeqnarray}
and also for \eqref{eq:RateEvamUpper:a} and \eqref{eq:RateEvamUpper:a6}:
\begin{IEEEeqnarray}{rCl}\label{eq:D2}
C_{m,i,k}^{\mathtt{ED}}(\mathbf{X}_i,\alpha_i) &\leq& C_{m,i,k}^{\mathtt{ED},(\kappa)}(\mathbf{X}_i,\alpha_i,\mu_{m,i,k}) \leq \Gamma_{i,k}^{\mathtt{D}}\ \text{and}\ \nonumber\\
C_{m,i,k}^{\mathtt{ED}}\bigl(\mathbf{X}_i^{(\kappa)},\alpha_i^{(\kappa)}\bigr) &=& C_{m,i,k}^{\mathtt{ED},(\kappa)}\bigl(\mathbf{X}_i^{(\kappa)},\alpha_i^{(\kappa)},\mu_{m,i,k}^{(\kappa)}\bigr).
\end{IEEEeqnarray}
From \eqref{eq:changeOP3:a}, it follows that
\begin{IEEEeqnarray}{rCl}\label{eq:D3}
&&\eta^{(\kappa+1)}\triangleq\underset{(i,k)\in\mathcal{D}}{\min}\Bigl\{C_{i,k}^{\mathtt{D}}\bigl(\mathbf{X}_i^{(\kappa+1)},\alpha_i^{(\kappa+1)}\bigr) \nonumber\\&&\qquad\qquad\qquad - \underset{m\in\mathcal{M}}{\max}\,C_{m,i,k}^{\mathtt{ED}}\bigr(\mathbf{X}_i^{(\kappa+1)},\alpha_i^{(\kappa+1)}\bigl)\Bigr\}\nonumber\\
&&\quad\geq\underset{(i,k)\in\mathcal{D}}{\min}\Bigl\{C_{i,k}^{\mathtt{D},(\kappa)}\bigl(\mathbf{X}_i^{(\kappa+1)},\alpha_i^{(\kappa+1)}\bigr) \nonumber\\
&&\qquad\qquad\qquad- \underset{m\in\mathcal{M}}{\max}\,C_{m,i,k}^{\mathtt{ED},(\kappa)}\bigl(\mathbf{X}_i^{(\kappa+1)},\alpha_i^{(\kappa+1)},\mu_{m,i,k}^{(\kappa+1)}\bigr)\Bigr\}\nonumber\\
&&\quad\geq\underset{(i,k)\in\mathcal{D}}{\min}\Bigl\{C_{i,k}^{\mathtt{D},(\kappa)}\bigl(\mathbf{X}_i^{(\kappa)},\alpha_i^{(\kappa)}\bigr) \nonumber\\
&&\qquad\qquad\qquad- \underset{m\in\mathcal{M}}{\max}\, C_{m,i,k}^{\mathtt{ED},(\kappa)}\bigl(\mathbf{X}_i^{(\kappa)},\alpha_i^{(\kappa)},\mu_{m,i,k}^{(\kappa)}\bigl)\Bigr\}\nonumber\\
&&\quad\stackrel{(a)}{=}\underset{(i,k)\in\mathcal{D}}{\min}\Bigl\{C_{i,k}^{\mathtt{D}}\bigl(\mathbf{X}_i^{(\kappa)},\alpha_i^{(\kappa)}\bigr) 
 - \underset{m\in\mathcal{M}}{\max}\,C_{m,i,k}^{\mathtt{ED}}\bigl(\mathbf{X}_i^{(\kappa)},\alpha_i^{(\kappa)}\bigr)\Bigr\}\nonumber\\
&&\quad\triangleq\eta^{(\kappa)}
\end{IEEEeqnarray}
where the equality $(a)$ is obtained by using the equalities in \eqref{eq:D1} and \eqref{eq:D2}. This result shows that the objective  is non-decreasing with  iteration number, i.e., $ \eta^{(\kappa+1)} \geq \eta^{(\kappa)}$. Also, it is clear that  $\bigl(\mathbf{X}^{(\kappa+1)},\boldsymbol{\alpha}^{(\kappa+1)}\bigr)$ is a better point for \eqref{eq:OP4} than $\bigl(\mathbf{X}^{(\kappa)},\boldsymbol{\alpha}^{(\kappa)}\bigr)$ whenever $\bigl(\mathbf{X}^{(\kappa+1)},\boldsymbol{\alpha}^{(\kappa+1)}\bigr)\neq\bigl(\mathbf{X}^{(\kappa)},\boldsymbol{\alpha}^{(\kappa)}\bigr)$. For  problem  \eqref{eq:OP4}, we have $\mathcal{V}^{(\kappa+1)}\supseteq\mathcal{V}^{(\kappa)}$ which is an immediate consequence and $\{\eta^{(\kappa)}\}_{\kappa \geq 1}$ is non-decreasing sequence. According to the Cauchy's theorem, the sequence $\{\eta^{(\kappa)}\}$ is bounded, i.e., $\underset{\kappa\rightarrow +\infty}{\lim}\eta^{(\kappa)} = \bar{\eta}$ with a limit point $\bar{\Psi}$. Then, each accumulation point $\Psi^{(\kappa)}$, if $\Psi^{(\kappa+1)}=\Psi^{(\kappa)}$, is a KKT point for \eqref{eq:OP5},  which obviously is also a KKT point for \eqref{eq:OP4} according to \cite[Theorem 1]{Marks:78}.  Proposition \ref{pro:2} is thus proved.

\section*{Appendix~E\\ \quad\  Proof of Lemma \ref{lemma:1} } \label{Appendix:E}
\renewcommand{\theequation}{E.\arabic{equation}}
\setcounter{equation}{0}
Under the assumption of the independence of Eves' channels,  constraint \eqref{eq:OP7:c} can be computed as
\begin{IEEEeqnarray}{rCl}\label{eq:E1}
&&\Pro\Bigl(\underset{m\in\mathcal{M}}{\max} C_{m,i,k}^{\mathtt{ED}}(\mathbf{X}_i,\alpha_i)\leq \Gamma_{i,k}^{\mathtt{D}}\Bigr) \geq \epsilon_{i,k} \nonumber\\
\Leftrightarrow\;&& \Pro\bigl( C_{m,i,k}^{\mathtt{ED}}(\mathbf{X}_i,\alpha_i)\leq \Gamma_{i,k}^{\mathtt{D}}\bigr) \geq \epsilon_{i,k}^{1/M}.
\end{IEEEeqnarray}
Note that the inequality \eqref{eq:E1} holds  easier if   Eves' channels are dependent since its RHS yields a smaller value. We further rewrite \eqref{eq:E1} based on the basic property of probability as
\begin{equation}\label{eq:E2}
 \eqref{eq:E1} \Leftrightarrow \Pro\bigl( C_{m,i,k}^{\mathtt{ED}}(\mathbf{X}_i,\alpha_i)\geq\Gamma_{i,k}^{\mathtt{D}}\bigr) \leq 1-\epsilon_{i,k}^{1/M}.
\end{equation}
Observe that the LHS of \eqref{eq:E2} is coupled between the optimization variables $(\mathbf{X}_i,\alpha_i)$ and the SCSI of the Eves. On the other hand, it requires an upper bound of the LHS of \eqref{eq:E2}, which is the outage probability for DL user ($i,k$). We make use of the well-known Markov inequality (MI) \cite{Billingsleybook,AkgunTCOM17}, i.e., $\Pro(Y \geq y) \leq \mathbb{E}\{Y\}/y$, to manipulate the LHS of \eqref{eq:E2} as
\begin{IEEEeqnarray}{rCl}\label{eq:E3}
&& \Pro\Bigl(\ln\Bigl(1 + \frac{\|\mathbf{H}_{m}^H\mathbf{w}_{i,k}\|^2}{\psi_{m,i,k}(\mathbf{X}_i)} \Bigr) \geq \alpha_i\Gamma_{i,k}^{\mathtt{D}} \Bigr)\qquad\qquad\qquad\\
 &&\quad =  \Pro\Bigl(\|\mathbf{H}_{m}^H\mathbf{w}_{i,k}\|^2  + \bigl(1-e^{\alpha_i\Gamma_{i,k}^{\mathtt{D}}}\bigr)\psi'_{m,i,k}(\mathbf{X}_i) \nonumber\\
&&\qquad\qquad\qquad\qquad\qquad\;  \geq  \bigl(e^{\alpha_i\Gamma_{i,k}^{\mathtt{D}}}-1\bigr) N_{e,m}\sigma^2  \Bigr)
\end{IEEEeqnarray}
\begin{IEEEeqnarray}{rCl}
&&\quad\stackrel{\mathtt{MI}}{\leq}\frac{\mathbb{E}\bigl\{\mathbf{w}_{i,k}^H\mathbf{H}_{m}\mathbf{H}_{m}^H\mathbf{w}_{i,k}+ \bigl(1-e^{\alpha_i\Gamma_{i,k}^{\mathtt{D}}}\bigr)\psi'_{m,i,k}(\mathbf{X}_i)\bigr\}}{\bigl(e^{\alpha_i\Gamma_{i,k}^{\mathtt{D}}}-1\bigr) N_{e,m}\sigma^2}\qquad\\
&&\quad = \frac{\mathbf{w}_{i,k}^H\bar{\mathbf{H}}_{m}\mathbf{w}_{i,k}+ \bigl(1-e^{\alpha_i\Gamma_{i,k}^{\mathtt{D}}}\bigr)\bar{\psi}_{m,i,k}(\mathbf{X}_i)}{\bigl(e^{\alpha_i\Gamma_{i,k}^{\mathtt{D}}}-1\bigr) N_{e,m}\sigma^2}\label{eq:E6}
\end{IEEEeqnarray}
where $\psi_{m,i,k}'(\mathbf{X}_i) = \psi_{m,i,k}(\mathbf{X}_i) - N_{e,m}\sigma^2$, and $\bar{\psi}_{m,i,k}(\mathbf{X}_i)$ is the expectation of $\psi_{m,i,k}'(\mathbf{X}_i)$ w.r.t. \eqref{eq:CDI}.  By replacing  the LHS of \eqref{eq:E2} with \eqref{eq:E6} and  after some straightforward manipulations, we arrive at \eqref{eq:OP7:c1}. It can be shown in a similar manner that \eqref{eq:OP7:d}  is converted to \eqref{eq:OP7:d1}, and thus the proof is completed.

\balance
\bibliographystyle{IEEEtran}
\bibliography{IEEEfull}
\end{document}